\documentclass[12pt]{article}
\usepackage{amsfonts,amsmath,amsthm,array,amssymb}
\newtheorem{theorem}{Theorem}
\newtheorem{lemma}{Lemma}

\usepackage{graphicx}
\usepackage{picinpar,graphicx}
\usepackage{graphics}
\usepackage{multicol}
\usepackage{tabularx}
\usepackage{ctable}
\usepackage{booktabs}
\usepackage{url}
\usepackage{graphicx}
\usepackage{epsfig}
\usepackage{graphics}
\usepackage{amsmath}
\usepackage[mathscr]{euscript}
\usepackage{ amssymb }
\usepackage{epstopdf}
\usepackage{booktabs}
\usepackage{setspace}

\DeclareGraphicsExtensions{.pdf,.png,.jpg,.eps}
\begin{document}

\title{Dynamics and Control of an Invasive Species: The Case of the Rasberry crazy ant Colonies}
\author{Valerie Cheathon$^{1}$, Agustin Flores$^{2}$, Victor Suriel$^{3}$, Octavious Talbot$^{4}$, \\ Dustin Padilla$^{5}$ ,Marta Sarzynska$^{6}$, Adrian Smith$^{7}$, Luis Melara$^{8}$}
\date{}
\maketitle
\begin{center}
\footnotesize $^{1}$ New College of Interdisciplinary Arts and Sciences, Arizona State University, Glendale, Arizona\\
\footnotesize $^{2}$ Department of Mathematics, Northeastern Illinois University, Chicago, Illinois\\
\footnotesize $^{3} $ Department of Mathematics, SUNY-Stony Brook, Long Island, New York\\
\footnotesize $^{4}$ Department Mathematics, Morehouse College, Atlanta, Georgia\\
\footnotesize $^{5}$ Mathematical Institute, University of Oxford, Oxford, UK\\
\footnotesize $^{6}$Applied Mathematics for the Life and Social Sciences, Arizona State University, Tempe, Arizona\\
\footnotesize $^{7}$Department of Mathematics, Shippensburg University, Shippensburg, Pennsylvania
\end{center}

\begin{abstract}

{This project is motivated by the costs related with the documented risks  of the introduction of non-native invasive species of plants, animals, or pathogens associated with travel and international trade. Such invasive species  often have no natural enemies in their new regions. The spatiotemporal dynamics related to the invasion/spread of {\it Nylanderia fulva}, commonly known as the Rasberry crazy ant, are explored via the use of models that focus on the reproduction of ant colonies. A Cellular Automaton (CA) simulates the spatially explicit spread of ants on a grid. The impact of local spatial correlations on the dynamics of invasion is investigated numerically and analytically with the aid of a Mean Field (MF) model and a Pair Approximation (PA)  model, the latter of which accounts for adjacent cell level effects. The PA model approach considers the limited mobility range of  {\it N. fulva}, that is, the grid cell dynamics are not strongly influenced by non-adjacent cells. 
The model determines the rate of growth of colonies of {\it N. fulva} under distinct cell spatial architecture. Numerical results and qualitative conclusions on the spread and control of this invasive ant species are discussed.}
\end{abstract}

\section{Introduction}\
\indent {\it Nylanderia fulva}, commonly known as the Rasberry crazy ant, is a species originating from Brazil and Argentina \cite{gotzek2012importance}. The first documented sighting of this species in the USA was in Pasadena, Texas in 2002, by local pest exterminator Tom Rasberry, hence the nickname, The Rasberry crazy ant \cite{horn2010examining,meyers2008identification,rasberryresponse}.  The morphological and behavioral similarities between ants of a certain genus can lead to the mislabeling of the species \cite{lapolla2011monograph}. This lack of global synthesis \cite{lapolla2011monograph} is apparent in the identification of the {\it N. fulva} and {\it N. pubens} (Caribbean ant) found in Florida. In \cite{zhao2012molecular}, the authors demonstrated a molecular approach that determined that Rasberry and Carribbean ants are the same. In \cite{lapolla2011monograph}, global synthesis is encouraged to reduce the number of mislabeled ant species, so that this particular species can be accurately accounted for. The lack of research on this species means that there is less information to use in estimating  parameters in the models.

\indent Since 2002, the Rasberry ant has spread to over 20 counties in both Texas and Florida, with recorded sightings in Louisiana and Mississippi \cite{gotzek2012importance}. Rasberry ants have caused up to 150 million dollars of damage to electrical equipment by chewing through insulation and shorting circuits \cite{meyers2009ant}. Electrocuted ants release a chemical that attracts other ants to the electrocuted ants location \cite{meyers2009ant,doe}. This sudden invasion inside of the electrical unit, leads to overheating and mechanical failures \cite{homeinvasive}. Rasberry ants can seemingly nest wherever there is a suitable food source. They have been known to inhabit homes, buildings and even bee hives \cite{homeinvasive}. The ants prefer warmer temperatures and colony numbers increase during the spring, summer and fall months \cite{drees2009rasberry}. According to \cite{doe}, Rasberry ants have not been observed engaging in nuptial flights. In other ant species, a nuptial flight involves a winged ant male and female, referred to as sexuals or reproductives, mating at or near the edge of the original nest. Though winged males are born in Rasberry ant populations, their ability to spread is attributed to their seasonal budding process \cite{doe}. 

\indent The Rasberry ant tends to form a group of nests where ants do not exhibit mutual aggression, known as a super colony \cite{aguillard2011extraction,Arcila2002,meyers2008identification}. Adjacent colonies can merge and some colonies reach population densities up to one hundred times larger than other ant populations \cite {lebrunimported}. In addition, their colonies are polygynous, where one colony can have multiple queens laying eggs \cite{Arcila2002}. In a personal interview, Tom Rasberry described an observed Rasberry ant nest consisting of  96 queens in a volume of dirt less than ten gallons \cite {Rasberry_interview}. 

Due to the sheer size of the colony population, the Rasberry ant is a potent competitor with the surrounding ant species, often displacing the original ants in the area \cite{jouvenaz1983natural}. In a study on the distribution of ants in Deer Park, Texas the percentage of collected specimen of {\it Solenopsis invicta}, or the Fire ant ( an invasive species in the area) dropped from 54.5\% of collected ant species in 2006 to 45.4\% in 2007. In the same time span, the percentage of Rasberry ants collected grew from 45.4\% to 61.8\% \cite{meyers2008identification}. Laboratory studies have shown that in one-on-one confrontations with Fire ants, Rasberry ants are at a disadvantage \cite{horn2010examining}. However, the sheer colony size of the Rasberry ant effectively allows them to control food resources in a neighboring region, which results in the displacement of the Fire ant \cite{horn2010examining,jouvenaz1983natural}. The Rasberry ant displaces other ant species by chemical attacks, nest raiding, outnumbering other ant species at food sources and diel foraging activity \cite{mcdonald2012investigation}.

\indent Rasberry ants diet is varied making them an omnivorous species \cite {mcdonald2012investigation}.  They will feast on arthropods and small vertebrates for their protein needs and the sweet nectar from mealybugs and aphids and other hemipterans, for their carbohydrate needs \cite{sharma2013honeydew}. They are also attracted to sweet parts of plants and fruit \cite{mcdonald2012investigation,sharma2013honeydew}. In their native lands, Rasberry ants are controlled, if not displaced by other ant species. However none of these species exist in the United States  \cite {piva2011ant}. Chemical control is also limited. Studies show that contact insecticides that contain chlorpyrifos are effective for 43 days after application \cite{mcdonald2012investigation}. Certain fipronil based pesticides such as Termidor ® SC and Taurus® SC are used as perimeter sprays around structures and homes \cite{mcdonald2012investigation}. The long lasting residual effects of fipronil-based products are ideal for providing areas of defense against immediate recolonization. Studies show that after a 3.4kg fipronil application, 18 weeks passed before healthy colonies began to re-infest \cite {1998}. Despite its success in terminating the Rasberry ant, environmental ramifications of the chemical are discouraged due to their effects on bees, mammals and other vertebrae and invertebrae \cite{tingle2000health}.


\noindent In order to study the spreading dynamics of Rasberry ants, we present three different approximation techniques, including both spatially independent and dependent models. These results will be used to characterize the spread of the Rasberry ant in the presence of pesticide, over a homogeneous environment. First, we construct a Cellular Automaton (CA) simulation. The CA is a stochastic computer simulation of the Rasberry ants spread in a homogeneous environment. Although it is the most spatially explicit of the three models, as a discrete space model it is computationally intensive because large numbers of trials are needed to lessen stochastic anomalies. The main advantage is that it incorporates the aggregate interactions of the cells representing the states of land in our grid, making it the most accurate of the models. However, rigorous mathematical analysis of our CA model is a daunting task, and the mathematical framework needed to do so is far beyond the scope of this report. Next, we construct the Pair Approximation (PA) model. The PA utilizes ordinary differential equations, and is a deterministic representation of the CA. In the PA we incorporate some degree of spatial dependence, which yields insight into the spread of the Rasberry ant \cite{hiebeler2000populations}. It requires less computing power, but mathematical analysis can be difficult because of the size of the ODE system. Finally, we construct a Mean Field model (MF). Although the MF is not as accurate as the PA, the simplicity of the MF model allows for intricate analysis. The MF assumes spatial independence, and spatial correlations neither exist nor develop.



\section{Mathematical  Models}

\indent We define three different states in a homogeneously mixed grid: $[0],[1],[2]$. The state containing pesticide and no ants is $[0]$, the state containing neither ants nor pesticide (empty land) is $[1]$, and the state containing ants and no pesticide is $[2]$. Let $\mu$ be the rate of application of pesticide on ant colonies per week.  New areas are colonized by the ants at a rate of $\phi$ per week,and pesticide degrades at a rate of $\epsilon$ per week. The states and parameters are summarized in \ref{table:variables}. We assume only land containing ants is sprayed with pesticide and that the pesticide is 100\% effective against the Rasberry ant.

The MF approximation provides a simplified view of the spread of the Rasberry ants. Spatial independence in the MF approximation suggests that local interactions of ant colonies are neglegent in interpreting succesful growth and budding. This implication can be disadvantageous because many local environnmenal impediments to the spread of an ant colony are eliminate with such a broad perspective. In the CA approximation, we assume a von Neumann neighborhood where new colonies are being born in the four cardinal directions adjacent to the mother colony. The stochastic and spatially explicit features of the CA make it ideal for exploring the effect of the random behavior on the spread of the Rasberry ant.

\subsection{Cellular Automaton Model}
\indent The stochastic simulation of the spread of ants on a grid was done using a CA simulation.  The occurrence of
an event  is based on a Poisson process in which the time between any two sequential events is exponentially distributed. We construct an $n \times n$ lattice of cells in which each cell is at a certain state $[i]$, for $i = 0,1,2$. The lattice has periodic boundary conditions (i.e. torus). An example of such a lattice is shown in the \ref{lattice-example}.
\begin{figure}[h!]
\begin{center}
\caption{\footnotesize An example of a lattice. Here $P[0]$ = 4/25, $P[1]$ = 11/25, $P[2]$ = 10/25 and $P[12]$ = 22/25}
\includegraphics[scale=.35]{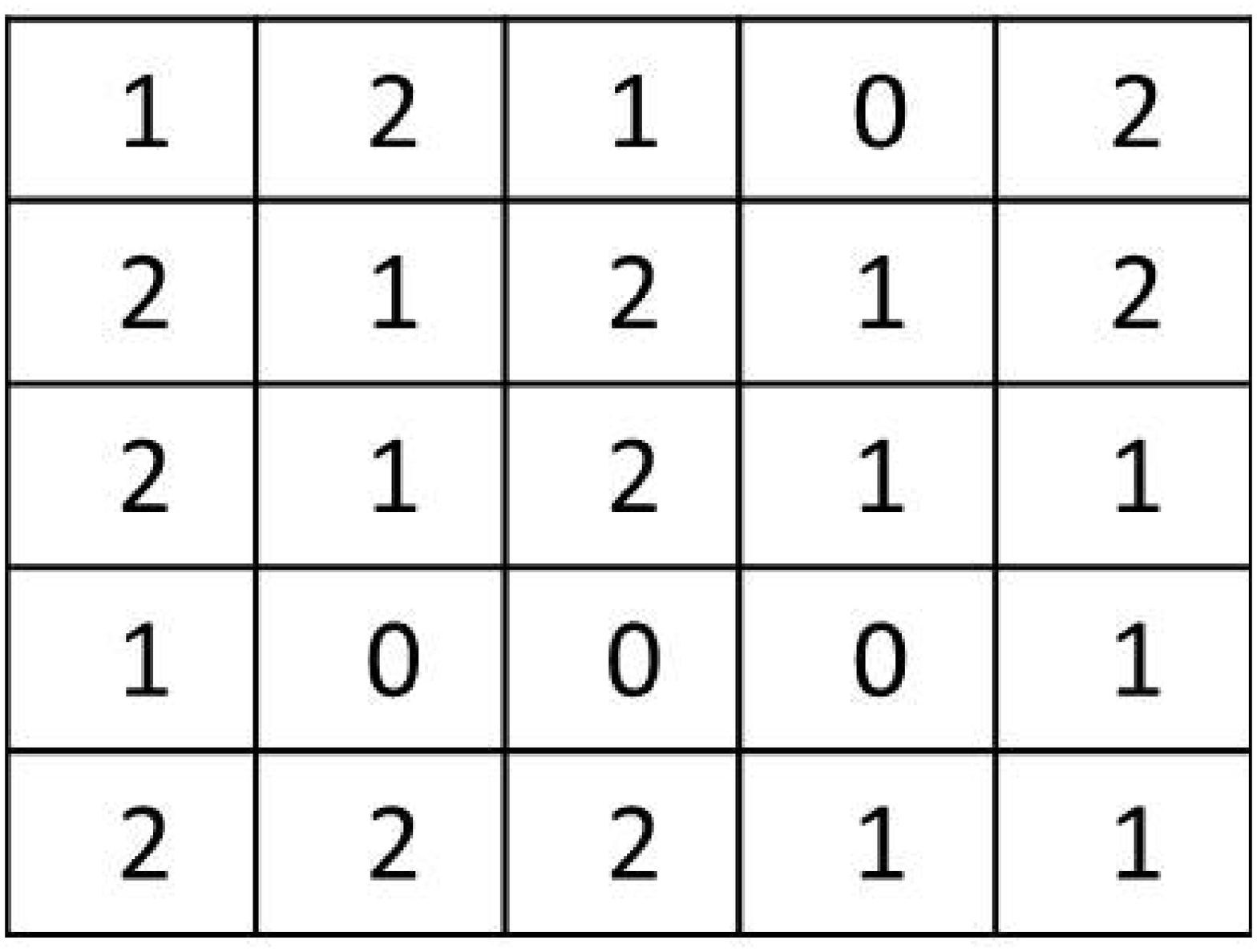}
\label{lattice-example}
\end{center}
\end{figure}
\noindent Before the random event process begins, a bookkeeping of the current state of all the cells is implemented and the result is represented by the indices of a matrix in our simulation. The only event that depends on both the cell acted on and its neighbors is colony birth. There is a count of all the $[12]$ pairs, which are necessary for the birth to occur.  A selection of three events based on the probability measure space created from the rate of each event and the total events rate. The event rates are defined in \ref{table:variables}.
%
%
%
\begin{table}[h!] 
\caption{\footnotesize Cellular Automaton State Variables and Event Rates }
\vspace {.5mm}
\centering 
\label{table:variables} 
\begin{tabularx}{0.83\textwidth}{ll}
\toprule
Variable/Event & Definition \\ 
\toprule
$S_{0}$ & Number of cells in state $[0]$ (Pesticide)\\ [0.5ex]
$S_{1} $ &Number of cells in state $[1]$ (Empty land ) \\ [0.5ex]	
$S_{2}$ & Number of cells in state $[2]$ (ant colony) \\ [0.5ex]
$S_{21}$ &Number of possible birth site pairs of state $[1]$ and $[2]$ \\ [0.5ex]
$\epsilon S_{0}$ & Degradation of pesticide \\[0.5 ex] 
$\mu S_{2}$ & Pesticide application \\[0.5 ex]
$ \frac{3}{4} \phi S_{21} $ & ant colony birth\\
$\epsilon S_{0}+\mu S_{2}+\frac{3}{4} \phi S_{21}$ & Any defined event occurs (total rate) \\
\bottomrule
\end{tabularx}
\end{table}

\begin{table}[h!]
\caption{\footnotesize Event Probabilities }
\label{table:eventprobs} 
\centering 
\begin{tabularx}{0.43 \textwidth}{ll}
\toprule
Probability  & Event \\ [0.5ex]
\toprule
$\frac{\epsilon S_{0}}{T_{R}}$ & Pesticide degradation\\ [0.5ex]
\vspace{0.4mm}
$\frac{\mu S_{2}}{T_{R}}$ & Pesticide application\\ [0.5ex]	
\vspace{0.2mm}
$\frac{\frac{3}{4} \phi S_{21}}{T_{R}}$ & ant colony budding \\ [0.5ex]
\bottomrule
\end{tabularx}
\end{table}

In terms of the stochasticity involved in the simulation, the event probabilities are calculated directly from the event rates and used to sample from the exponential distribution, with the total rate as the parameter. \ref{table:eventprobs} outlines the probability of each event occurring. We denote the total rate by $T_{R}$. The time between the last event is sampled from the exponential distribution as follows: 
 \begin{eqnarray*}
 t \sim  {\rm exp}\left({\left(T_{R}\right)t }\right).
\end{eqnarray*}
After each event the rates and probabilities are recalculated and a new event is randomly selected from the updated probability interval. The result is that the time between any two events is different,  as is to be expected from our Poisson waiting process assumption. 
An event is selected using a uniform number sample between 0 and 1 which falls somewhere on the probability interval constructed, as shown in Figure 2.
\begin{figure}[h!]
\begin{center}
\includegraphics[scale=.35]{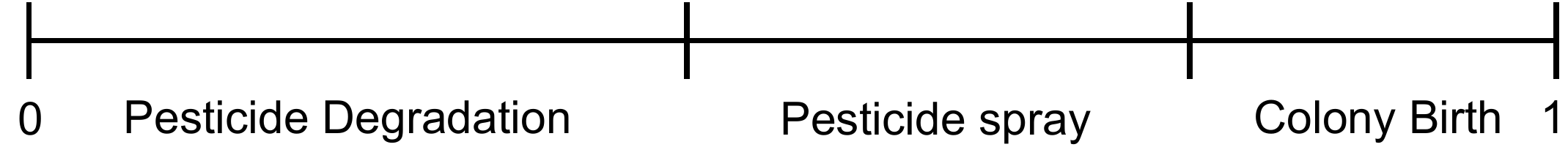}
\caption{\footnotesize  A generalization of a constructed probability interval}
\label{Measure Space}
\end{center}
\end{figure}
 A cell (matrix element) of the appropriate state is selected based on the event to occur.  For example, if a pesticide application is selected as an event, a cell in the ant colony state $[2]$ will be selected and it will be sprayed (i.e. changed to pesticide state $[0]$). The only event that can be "wasted" is a colony birth, in which there are no neighboring cells of state $[1]$ of a cell in state $[2]$ for the colony birth to occur. After each event, the change in a state is reflected in the matrix representing the lattice and there is a new bookkeeping of cell states for the current time step. The numerical results generated from the simulation are used for comparative analysis with  the MF and PA models  and will be discussed in the numerical analysis .  \vspace{3mm}

\noindent The following rules are implemented:

\begin{itemize}

\item \small {Events are chosen probabilistically from the rates and total rate}

\item \small {Pesticide degradation only depends on the random selection of a cell in state $[0]$}

\item\small {Pesticide application (spraying) only depends on the random selection of a cell in state $[2]$ which is changed to state $[0]$ }

\item\small {If a colony birth event is selected, then a cell in state $[2]$ is randomly selected  from its four cardinal neighbors as possible new sites for a birth. The first neighbor meeting this criteria will have its state changed to $[2]$. A colony birth event is wasted if there  are no neighbors of the selected state $[2]$ cell in state $[1]$ (i.e no empty land).}

\item\small { Waiting time between any two events is different because the program samples from the exponential distribution with a parameter which changes after each iteration.}
\end{itemize}

\noindent In \ref{BirthProcess} there is  an example of a birth process selected and completed successfully.
\begin{figure}[h!]
\centering
\includegraphics[width=0.90\textwidth,scale=.10]{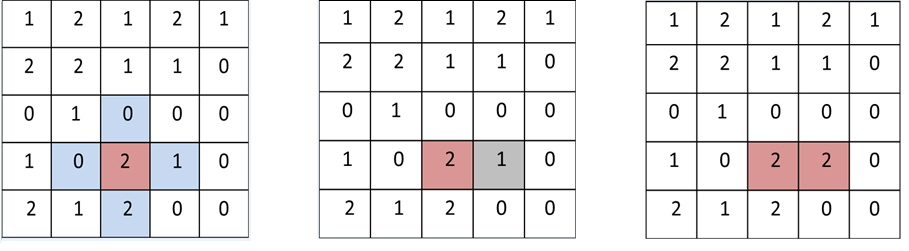}
\label{BirthProcess}
\caption{\footnotesize (a): A birth process is chosen and a cell in state $[2]$ is randomly selected. (b): Starting with the northern neighbor of state $[2]$, the program searches clockwise and selects the cell to the right in state $[1]$. (c): The cell in state $[1]$ receives a new ant colony and changes to $[2]$ }
\end{figure}


\subsection{Mean Field Model}
In the MF model, we assume spatial independence: the probability of finding any cell in a particular state is independent from its immediate neighbors, and only depends on the overall makeup of the lattice. These are the equations for the MF model, where
the terms are explained in table (\ref{table:MFStateVar}):

\begin{table}[!h]
\centering
\caption{\footnotesize State Variables in the MF Model}
\label{table:MFStateVar} 
\centering 
\begin{tabularx}{0.73\textwidth}{ll}
\toprule
Variable/Parameter  & Definition \\ [0.5ex]
\toprule
\\
$P[0]$ & Proportion of land at state 0 (Pesticide) \\ [0.5ex]
$P[1]$ & Proportion of land at state 1 (Empty land) \\ [0.5ex]	
$P[2]$& Proportion of land at state 2 (ant colony) \\ [0.5ex]

$\epsilon$ & Pesticide degradation rate \\[0.5 ex] 
$\mu$ & Pesticide application rate \\[0.5 ex]
$\phi$ &  Ant colony birth\\[0.5 ex]
\bottomrule
\end{tabularx}
\end{table}
\begin{eqnarray}
\frac{d}{dt}P[0]& = &\mu P[2] - \epsilon P[0] \label{Equation1}
 \\
\frac{d}{dt}P[1]& = &\epsilon P[0] - \phi P[1] P[2] \label{Equation2}
 \\
\frac{d}{dt}P[2]& = &\phi P[1] P[2] - \mu P[2] \label{Equation3}
\end{eqnarray}

\noindent The parameters $\mu$, $\epsilon$, and $\phi$ are positive constants, and ${P[0] + P[1] + P[2] = 1}$ by the law
of total probability.  Substitution for ${P[0] = 1- P[1] - P[2]}$  reduces our system of ODEs \ref{Equation1}--\ref{Equation3} to:
\begin{eqnarray}
\frac{d}{dt}P[1]& = &\epsilon \left(1 - P[1] - P[2]\right) -\phi P[1] P[2] \label{Equation4}
 \\
\frac{d}{dt}P[2]& = &\phi P[1] P[2] - \mu P[2] \label{Equation5}.
\end{eqnarray}

\noindent There are two equilibria points. The Ant Free Equilibrium (AFE) is at ${(P[1]^*,P[2]^*)=(1,0)}$, and the Ant Endemic Equilibrium (AEE) is $ {\left(\frac{\mu}{\phi},
\frac{\epsilon}{\mu + \epsilon}(1- \frac{\mu}{\phi}) \right)}$. Note that when $\mathscr{I}_0 < 1$, the $P[1]$ coordinate
of the AEE is greater than 1. Because the lattice proportions of our states must between 0 and 1 this is not a biologically relevant situation, and we do not consider the AEE when $\mathscr{I}_0<1$.

\subsubsection{MF Local Stability Analysis}
\begin{theorem}
\emph{(Local stability of AFE)}
If $\mathscr{I}_0 \equiv \frac{\phi}{\mu} < 1$, then the AFE is locally asymptotically stable. If $\mathscr{I}_0 > 1$, then
the AFE is a saddle point and unstable.
\end{theorem}
\begin{proof}
The Jacobian of system \ref{Equation4}--\ref{Equation5} at the AFE is
\begin{equation}
J(AFE) = \begin{bmatrix}
-\epsilon &-\epsilon -\phi \\
0 & \phi - \mu
\end{bmatrix}.
\end{equation}
\noindent This is an upper triangular matrix with eigenvalues $-\epsilon$ and $\phi-\mu$. From the second eigenvalue we define the {\it basic invasion number}
\[  \mathscr{I}_0 \equiv \frac{\phi}{\mu}.  \]
$\mathscr{I}_0$ measures the average number of colonies budded from a mother ant colony during its lifespan in land of mostly empty cells. If $\mathscr{I}_0 < 1$, our two eigenvalues are negative and the AFE is locally
asymtotically stable. In this instance, ant colonies are being sprayed faster than they are reproducing. Therefore, ants become extinct. 

If $\mathscr{I}_0>1$, one eigenvalue is positive and the AFE is a saddle point. The initial colonies are reproducing faster than they can be sprayed, and they will survive the initial colonization phase.
\end{proof}
\noindent Next, we prove the local stability of the AEE = $ {\left(\frac{\mu}{\phi}, \frac{\epsilon}{\mu + \epsilon}(1- \frac{\mu}{\phi}) \right)}$.T
\begin{theorem}
\emph{(Local stability of AEE)}
If $\mathscr{I}_0 > 1$, then the AEE is locally asymptotically stable.
\end{theorem}
\begin{proof}
The Jacobian of system \ref{Equation4}--\ref{Equation5} at the AEE is:
\begin{equation}
J\left(\frac{\mu}{\phi},\frac{\epsilon}{\mu + \epsilon}(1- \frac{\mu}{\phi}) \right) = \begin{bmatrix}
-\epsilon  -\frac{\epsilon(\phi - \mu)}{\mu + \epsilon} & \epsilon-\mu \\
\frac{\epsilon (\phi - \mu)}{\mu + \epsilon} & 0
\end{bmatrix}
\end{equation}
with eigenvalues:
\[ \lambda_{\pm} =   \frac{-\epsilon^2 -\phi \epsilon \pm \sqrt{\epsilon^4 +2\epsilon^2(2\mu-\phi) + 8\epsilon^2\mu(\mu-\phi) + \epsilon^2 \phi^2 + 4 \epsilon\mu^2(\mu-\phi)}}{2(\mu + \epsilon)}.\]
$\lambda_-$ always has a negative real part because it is the difference between a negative real number and a radical term, which is either a non-negative real number or a pure imaginary number.
To  show that $\lambda_+$ has a negative real part, we look at two cases. First, we consider the case where the radicand is negative. Under
this condition, the radical term is purely imaginary and $\lambda_+$ is a complex number with a negative real part. If the radicand is non-negative, the radical term is a non-negative real number. Remember that when the AEE is biologically relevant,  $\mathscr{I}_0 > 1$ and $\phi > \mu$. Therefore we have that:
\[\sqrt{\epsilon^4 +2\epsilon^2(2\mu-\phi) + 8\epsilon^2\mu(\mu-\phi) + \epsilon^2 \phi^2 + 4 \epsilon\mu^2(\mu-\phi)}<\sqrt{\epsilon^4 +2\epsilon^2\phi + \epsilon^2 \phi^2}=\epsilon^2 + \phi\epsilon \nonumber.\]
Therefore,
\begin{equation}
\lambda_{+} < \frac{-\epsilon^2 - \phi \epsilon + (\epsilon^2 + \phi \epsilon)}{2(\mu + \epsilon)} = 0.
\end{equation}
since $\lambda_{\pm}$ both have negative real parts, the AFE is asymptotically stable.
\end{proof}
\subsubsection{MF Global Stability Analysis}

In this section we prove some theorems regarding global stability of our equilibria points. In $\mathbb{R}^2$, $P[1]$ is on our $x$
axis and $P[2]$ is on our $y$ axis. Let $\mathscr{D}$ be the closed region bounded by the triangle with vertices $(0,0)$, $(0,1)$, and $(1,0)$, as shown in figure \ref{MFmodelphaseplane}. 
\begin{figure}[h!]
\centering
\includegraphics[scale=0.6]{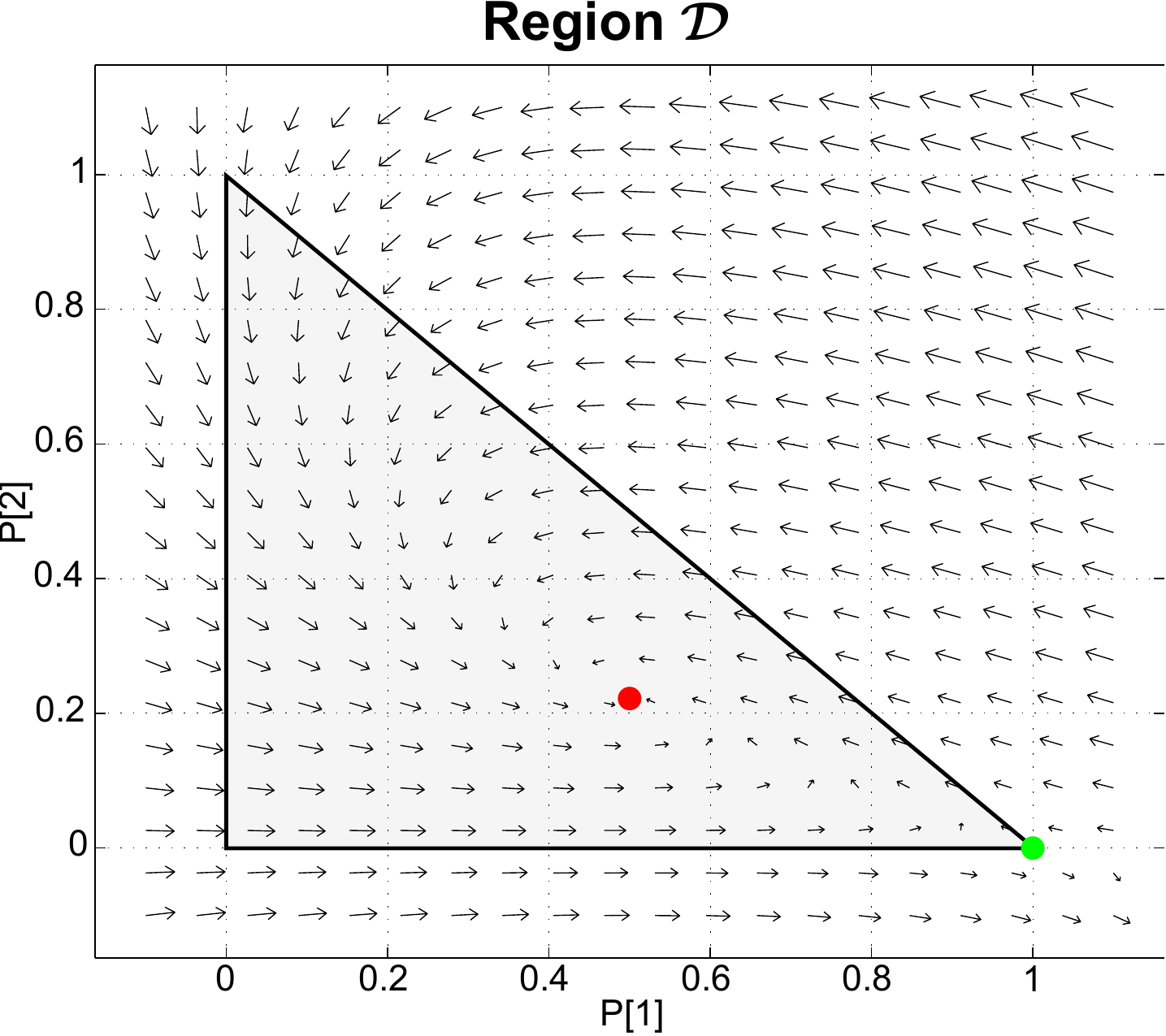}
\caption{\footnotesize The phase plane for our MF model is the triangle $\mathscr{D}$, with the AFE (green) and the AEE (red).}
\label{MFmodelphaseplane}
\end{figure}
Anything outside of  $\mathscr{D}$ is irrelevant because our state variables need to be positive and $P[1] + P[2] \leq 1$. Let $\partial\mathscr{D}$ denote the boundary of $\mathscr{D}$, and $\mathscr{D}^0$ denote the
interior of $\mathscr{D}$. We first prove the global stability of the AEE.
\begin{theorem}
\emph{(Global stability of the AEE)}
\label{globalAEE}
If $\mathscr{I}_0 > 1$, then the AEE is globally asymptotically stable on $\mathscr{D}^0$.
\end{theorem}
\begin{proof} In our two dimensional dynamical system, trajectories can only flow in one of three ways. They can either diverge (the trajectory is unbounded), they can converge to a stable limit cycle, or they can flow to a stable equilibrium point. To show that no trajectories diverge, we show that $\mathscr{D}$ is positive invariant under equations ($ \ref{Equation4}-\ref{Equation5}$). Consider the vector flow of our equations on $\partial \mathscr{D}$.  On the bottom triangle side
$B_1=\{\partial \mathscr{D}: P[2] = 0$\}, we have
\begin{eqnarray*}
\frac{d}{dt}{P[1]} &=& \epsilon\left  ( 1-P[1] \right) \geq 0 
\\
\frac{d}{dt}{P[2]} &=& 0.
\end{eqnarray*}
$P[2]$ has a constant time derivative, so it remains zero for all time. It follows that the trajectory with initial conditions (0,0) flows to the AFE at (1,0) along $B_1$, and $B_1$ is a trajectory in $\mathscr{D}$. Therefore, no other trajectories in 
$\mathscr{D}$ can
intersect $B_1$. On the left side of the triangle $B_2 = \{\partial \mathscr{D} : P[1] = 0\}$, 
\begin{equation*}
\frac{d}{dt} P[1] = \epsilon P[0] \geq 0,
\end{equation*}
so the vector field on $B_2$ points to the right and into $\mathscr{D}$. Along the side ${B_3=\{\partial \mathscr{D} : P[1]+P[2]=1\}}$ our time derivatives are:
\begin{eqnarray*}
\frac{d}{dt} P[1 ]&=& -\phi P[2]\left(1-P[2] \right)
\\
\frac{d}{dt} P[2] &=&\phi  P[2]\left(1-P[2]-\frac{\mu}{\phi} \right).
\end{eqnarray*} 
If our vector field is to point into $\mathscr{D}$, we must have that $\frac{d}{dt} P[1] + \frac{d}{dt} P[2] \leq 0$. It can 
readily be seen that 
\[ \frac{d}{dt} P[1] + \frac{d}{dt} P[2] =- \mu P[2] \leq 0. \]
Since our vector flow on $\partial \mathscr{D}$ points towards $\mathscr{D}$, it is a positively invariant set, and no trajectories in $\mathscr{D}$ can diverge.

 Next, we use the Bendixson-Dulac theorem to show there are no limit cycles in $\mathscr{D}^0$ when $\mathscr{I}_0 > 1$. The function $G \left(P[2] \right)= \frac{1}{P[2]}$ is a continuously differentiable function in  $\mathscr{D}^0$
Reorganizing the equations in system (\ref{Equation4})--(\ref{Equation5})
\begin{eqnarray*}
 \frac{d}{dt} P[1] &=& \epsilon \left(1-P[2] \right) - P[1] \left( \epsilon +\phi P[2] \right) 
\\
\frac{d}{dt} P[2] &=& P[2]\phi \left(P[1]- \frac{\mu}{\phi}\right).
\end{eqnarray*}
Now, applying Dulac's criterion with the function $G \left(P[2]\right) $,
\begin{eqnarray}
& &\frac{\partial}{\partial P[1]}\left( G \left (P[2] \right)\frac{d}{dt} P[1] \right) + \frac{\partial}{\partial P[2]}\left( G \left(P[2] \right)\frac{d}{dt} P[2] \right) \label{DulacCriterion}
\\ \nonumber
&=&-\frac{ \left(\epsilon +\phi P[2] \right)}{P[2]} <0. \nonumber
\end{eqnarray}
\noindent Equation (\ref{DulacCriterion}) is strictly negative, so by Dulac's Criterion, $\mathscr{D}^0$ has no periodic orbits. We have eliminated the possibility of divergence and convergence to a limit cycle, so any trajectory  in $\mathscr{D}_0$ must flow to a steady equilibrium point. The AEE is the only stable equilibrium point, so it is
globally stable.
\end{proof}
\indent We next use La Salle's invariance principle to show that when $\mathscr{I}_0 < 1$, the AFE is globally stable.
\begin{lemma}
\emph{(La Salle's invariance principle)}
Let $\Omega \subset \mathscr{D}$ be a compact and positive invariant set with respect to the system $\dot{\bf{x}} = f(\bf{x})$, where $f:\mathbb{R}^n \to \mathbb{R}^n$ is the function of time derivatives.
Let $V: \mathscr{D} \to \mathbb{R}$ be continuously differentiable with $\dot{V}(\mathbf{x}) \leq 0$ over $\Omega$. Let
$E$ be the set of states for which $\dot{V}(\mathbf{x}) = 0$. Let $M$ be the maximal positive invariant set in $E$. Then every
solution with initial state in $\Omega$ asymptotically approaches the set $M$.
\end{lemma}
\begin{theorem}
\emph{(Global stability of the AFE)}
If $\mathscr{I_0} < 1$, then the AFE is globally stable on $\mathscr{D}$.
\end{theorem}
\begin{proof}

We have already proved that $\mathscr{D}$ is positive invariant, and by the Hine-Borel theorem, it is compact because it is
closed and bounded.
Let  $V \left(P[1],P[2] \right) = P[2]$. We have that for $V$, 
\begin{eqnarray*}
V&\geq& 0 \qquad \mbox{when } (P[1],P[2])\in\mathscr{D}
\\
V  &=& 0 \qquad \mbox{when } P[2] = 0.
\\
\frac{d}{dt} V &=&  P[2]\phi \left(P[1]- \frac{\mu}{\phi} \right) < 0
 \end{eqnarray*}
where the last inequality holds because $\frac{\mu}{\phi} > 1$ when $\mathscr{I}_0 <1$. The bottom boundary edge $B_1$
where $V = 0$ is positively invariant, and the maximally positive invariant subset of $B_1$ is the AFE. From La Salle's
invariance principle, every trajectory starting in $\mathscr{D}$ flows to the AFE, and it is globally stable.
\end{proof}
\subsection{Pair Approximation Model}
In this section we explicitly derive the PA equations from \textit{a priori} assumptions about ant colonization. We start with
an infinite lattice $\mathscr{L}$ with cells in states [0] (pesticide), [1] (empty), and [2] (ants), with no boundary and make the following assumptions:
\begin{itemize}
\item Pesticide application and degredation of a cell is density independent and does not depend on any neighboring cell.
\item Only one cell can change its state at a time.
\item New ant colonies come from a single neighboring colony, and every new ant colony needs an empty cell (state [1]) to occupy.
\end{itemize}
 Let $[ij]$ be an adjacent pair of cells ordered left to right (or
north to south) in states $i$ and $j$. Given any random pair of adjacent cells on our grid, let $P[ij]$ be the probability that cells $[ij]$ are next to each other. By the symmetry of $\mathscr{L}$, $P[ij] = P[ji]$.
The term  $P[ij]$ is referred to as a state variable. All possible states are listed below: 
\begin{equation*}
\begin{matrix}
[00] & [01] & [02] \\
[10] & [11] & [12] \\
[20] & [21] & [22].
\end{matrix}
\end{equation*}
A von Neumann neighborhood will be used for the PA model; each cell is only in contact with cells in the 4 cardinal directions north, south, east, and west, as shown in figure \ref{Cross Diagram}
\begin{figure}[h!]
\centering
\includegraphics[width=2.0\textwidth,trim = 30mm 130mm 20mm 70mm, clip, width=10cm]{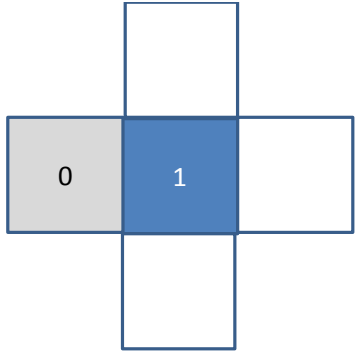}
\caption{\footnotesize A site of a $[21]$ pair and Von Neumann neighborhood around $[1]$.}
\label{Cross Diagram}
\end{figure}
\subsubsection{Development of the PA Model}
The rate of change of $P[ij]$ is the the rate of all states changing to $[ij]$ minus the rate of $[ij]$ changing. Written as a differential equation,
\[\frac{d}{dt}P[ij] = \sum inflow - \sum outflow\]
As an example, we derive the differential equation for the variable $P[01]$. From our assumption that only one cell can change
at a time, the only states that can change to state [01] are states [00] and [21]. The full set of transition rates are shown in figure (\ref{padevelop}) Cells in states [00] change to state [01] at a constant rate of $\epsilon$, and state [21] change to [01] at a rate of $\mu$. The [01] state can change to state [11] at a rate $\epsilon$, and can change to state [02] at a rate $\frac 34 \phi Q_{2|1}$, where 
\[Q_{2|1} =  \frac{P[21]}{P[01] + P[11] + P[12]}\]
 is the conditional probablity of finding a [21] state pair given that we have a cell in state [1].
\begin{figure}
\centering
\includegraphics[width=2.0\textwidth,trim = 70mm 70mm 70mm 70mm, width=5cm,angle=180] {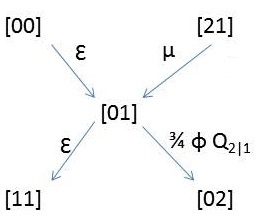}
\caption{The inflows and outflows of state [01].}
\label{padevelop}
\end{figure}
 Putting these rates together,
\begin{equation}
\frac{d}{dt}P[01] = \epsilon P[0 0] +\mu P[1 2] - P[01] \left(\epsilon + \frac{3}{4}\phi Q_{2|1} \right),
\end{equation}
Table \ref{table:variables}  lists all the state variables for the approximation model and the state change parameters.\\
\begin{table}[h!]
\caption{\footnotesize PA Model State Variable Definitions}
\label{table:variables} 
\centering 
\begin{tabularx}{\textwidth}{>{} lX}
\toprule
Variable/Parameter  & Meaning \\ [0.5ex]
\toprule
\\
$P[00]$ & Proportion of pairs of adjacent cells at state 0 (pesticide) \\ [0.5ex]
$P[01] \hskip.2cm \&  P[10]$ & Proportion of pairs at state 1 and state 0 (empty land \& pesticide) \\ [0.5ex]	
$P[11]$ & Proportion of pairs at state 1 (empty land) \\ [0.5ex]
$P[02] \hskip.2cm \& P[20]$ & Proportion of pairs at state 0 and state 2 (pesticide \& ant colony) \\ [0.5ex]
$P[12] \hskip.2cm \& P[21]$ & Proportion of pairs at state 1 and state 2 (empty land \& ant colony) \\[0.5ex]
$P[22]$ & Proportion of pairs at state 2 (ant colony) \\ [0.5ex]
$\epsilon$ & Pesticide degradation rate \\[0.5 ex] 
$\mu$ &  Pesticide application \\[0.5 ex]
$\phi$ & Ant colony birth rate\\[0.5 ex]
\bottomrule
\end{tabularx}
\end{table}
\noindent By the law of total probability,
\[ \sum^2_{j=0} \sum^2_{i=0} P[ij] = 1. \]
This relationship can be used to eliminate one of the state variables $P[ij]$. We elmininated the variable $P[00]$. Following a similar process for remaining variables, the PA model is given by
\begin{eqnarray*}
\frac{d}{dt}P[0 1] & = &\epsilon\big( 1 - 2P[01] - 2P[02] - 2P[12] -P[11] - P[22]\big) +\mu P[1 2] - P[01] \left(\epsilon + \frac{3}{4}\phi Q_{2|1} \right)
 \\
\frac{d}{dt}P[0 2 ] &= & P[01] \frac{3}{4}\phi Q_{2|1} + \mu P[2 2]- \mu P[02] -\epsilon P[02]
 \\
\frac{d}{dft}P[1 1] & = & 2\epsilon P[01]  - 2 P[1 1] \left(\frac{3}{4}\phi Q_{2|1} \right) 
\\
\frac{d}{dt}P[12] & = & \epsilon P[02]+P[11] \left(\frac{3}{4}\phi Q_{2|1} \right)-P[12] \left(\frac{3}{4}\phi Q_{2|1}+\frac{\phi}{4}+\mu \right)
 \\
 \frac{d}{dt}P[2 2] &=& 2P[1 2] \left( \frac{3}{4}\phi Q_{2|1} + \frac{\phi}{4} \right) - 2\mu P[2 2]
\end{eqnarray*}
with
\begin{equation}
P[00] = 1 - 2P[01] - 2P[02] - 2P[12] -P[11] - P[22]
\end{equation}
At the AFE for the PA model, no cells are in states [i0] or
[i2], thus we have:
\begin{equation}
\mbox{AFE} = (P[01]^*, P[02]^*, P[11]^*, P[12]^*, P[22]^*) = (0,0,1,0,0)
\end{equation}
\subsubsection{PA Local Stability Analysis}
We use a next generation operator to calculate the basic invasion number  $\bar{\mathscr{I}}_0$ for the PA model, and to analyze the local stability of the AFE. The vector X of infected classes is 
\begin{gather}
X =  \begin{bmatrix}
P[02]\\
P[12]\\
P[22]
\end{bmatrix}
\end{gather}

\noindent Using X, we have the next generation matrix $FV^{-1}$:
\begin{equation}
FV^{-1}=\begin{bmatrix}
 0 & 0  & 0 \\
\\
 \frac{3 \phi \epsilon}{\mu (\phi + 4 \mu +4\epsilon)} & 3\phi\frac{(\mu  +\epsilon)}{\mu (\phi + 4 \mu + 4\epsilon)} &\frac{3}{2} \frac{\phi \epsilon}{\mu (\phi + 4 \mu+ 4\epsilon)} \\
\\
 0 & 0 & 0
\end{bmatrix}
\end{equation}
The basic invasion number $\bar{\mathscr{I}}_0$ is the spectral radius of $FV^{-1}$, given by
\begin{equation}
\bar{\mathscr{I}}_0 = 3\phi\frac{(\mu  +\epsilon)}{\mu (\phi + 4 \mu + 4\epsilon)}.
\end{equation}
Furthermore, note
\begin{equation}
\bar{\mathscr{I}}_0 = 3\frac{\phi}{\mu}\frac{(\mu  +\epsilon)}{(\phi + 4 \mu + 4\epsilon)} = \left(  \frac{3}{\frac{\phi}{4(\epsilon + \mu)}+4}\right) \mathscr{I}_0 < \mathscr{I}_0.
\end{equation}

\noindent  $\bar{\mathscr{I}}_0$ is the average number of secondary colonies that a single initial colony will bud into before an application of pesticide kills it. Therefore, it is an indicator of the survival of ant colonies. Because $\bar{\mathscr{I}}_0 < \mathscr{I}_0$, our models predict that the initial ant population will
fare worse in the PA model than it will in the MF model. This is to be expected because the PA model incorporates spatial
dependence. There are two processes that are hindering ant growth. Not only does the pesticide
kill the ant at a rate $\mu$, but the ants cannot recolonize a cell that has been sprayed by pesticide until the pesticide degrades,
which happens at a rate $\epsilon$.
On the other hand, the MF model assumes spatial independence: every cell's behavior is independent of the state of its neighbors.
This can be observed in the lack of dependence of $\epsilon$ and $\mathscr{I}_0$. Although an ant cell can be directly sprayed by a pesticide, 
spatial independence assumes that colonizing colonies are not aware of whether or not an adjacent cell has been sprayed with pesticide.

A deeper biological understanding of the invasion number $\bar{\mathscr{I}}_0$ can be obtained by writing $\bar{\mathscr{I}}_0$ as following:
\begin{equation}
\bar{\mathscr{I}}_0 = \frac 34 \mathscr{I}_0 \left(\frac{\mu + \epsilon}{\frac{\phi}4 + \mu +\epsilon} \right)
\end{equation}
The first term $\frac 34 \mathscr{I}_0$ represents the ant's limited colonizing abilities in a von Neumann neighborhood. An ant colony that was created from a previous ant colony only has 3 out of the 4 surrounding cells to colonize, because the previous ant colony is still in a neighboring cell. Since every ant colony comes from a previous one, the ants can only colonize at at rate $3/4$ of the MF model. The term $(\mu + \epsilon)/(\frac{\phi}{4} + \mu + \epsilon)$ which we denote as the discount factor is a term that's always less than 1, and it represents the tendency for a given cell
to change out of states 0 and 2. It is the rate of pesticide application and degradation divided by the total rates of all events.
The ants need surrounding cells in the empty land state in order to create new colonies, so when the turnover rates of cells in states 0 and 1 are high (such as when $\epsilon$ is very large, and 0 states become 1 at a high rate), the discount  factor approaches 1 and the ants have an easier time invading new lands.
By solving for $\phi$, we can prove that $\bar{\mathscr{I}}_0 > 1$ if and only if
\begin{equation}
\phi > \frac{4\mu(\mu+\epsilon)}{2\mu+3\epsilon} \label{ineq}.
\end{equation}
Taking the partial derivative of the right hand side of (18) with respect to $\epsilon$, we have 
\begin{equation}
\frac{\partial}{\partial\epsilon}\frac{4\mu (\mu+\epsilon)}{2\mu+3\epsilon} = 
-\frac{4\mu^2}{(2\mu + 3\epsilon)^2}.
\end{equation}
The right hand side of inequality (\ref{ineq}) will decrease as the degradation rate of the pesticide $\epsilon$ increases. Biologically speaking, the ant colonies will better thrive as $\epsilon$ increases because cells
with pesticide will degrade faster into habitable cells, which allows the ants to recolonize those areas. The right hand side of (\ref{ineq}) reaches a maximum as $\epsilon$ approaches 0, in which case
\begin{equation}
\lim_{\epsilon \to 0} \frac{4\mu (\mu+\epsilon)}{2\mu+3\epsilon} = 2\mu.
\end{equation}
In other words, if $\phi \geq 2\mu$, then
\begin{equation}
\frac{4\mu (\mu+\epsilon)}{2\mu+3\epsilon}< 2\mu \leq \phi 
\end{equation}
and $\bar{\mathscr{I}}_0 > 1$. If the ants are colonizing at a rate twice as fast as the rate of pesticide application then the initial
colony invasion will persist if $\epsilon$ is strictly greater than $0$. Thus, the ant colonies will still find it easy to infiltrate an environment in which pesticide is being applied but eventually degrades.
The right hand side of (\ref{ineq}) reaches an asymptotic minimum as $\epsilon$ approaches infinity, in which case
\begin{equation}
\lim_{\epsilon \to \infty} \frac{4\mu (\mu+\epsilon)}{2\mu+3\epsilon} = \frac {4\mu}3.
\end{equation}
Therefore, if 
\begin{equation}
\phi <  \frac{4\mu}3,\qquad 
\end{equation}
then
\begin{equation}
 \qquad \phi <\frac{4\mu (\mu+\epsilon)}{2\mu+3\epsilon}
\end{equation}
and $\bar{\mathscr{I}}_0 < 1$.

\begin{eqnarray*}
\bar{{\cal I}}_0 =\left(  \frac{\mu + \epsilon}{\frac{\phi}{4} + \mu + \epsilon} \right) \left( \frac{\frac34 \phi}{\mu} \right)
\end{eqnarray*}
\begin{table} 
\caption{Nonlinear Model Results}   
\centering 							
\begin{tabular}{l l c}					
\toprule 
Limits & Biological Interpretation & Good or Bad for ant? \\
\midrule 
$\epsilon \to 0 \Rightarrow$ $\bar{{\mathscr{I}}}_0  \to \left( \frac {1}{{1} + \frac{1}{4}(\frac {\phi}{\mu})} \right) \left( \frac{\frac34 \phi}{\mu} \right) $&  Pesticide never degrades & Bad \\    \\
$\epsilon \to \infty \Rightarrow$ $ \bar{{\cal I}}_0  \to\frac{\frac34 \phi}{\mu} $ & Pesticide is ineffective & Good  \\ \\
$\mu \to 0 \Rightarrow$ $\bar{{\cal I}}_0  \to \infty $ & No pesticide application & Very Good \\ \\
$\mu \to \infty \Rightarrow$$\bar{{\cal I}}_0  \to 0$ & Pesticide applied continuously & Very Bad \\
\bottomrule                                                               
\end{tabular}
\label{table:nonlin}                                          
\end{table}
\indent Table \ref{table:nonlin} shows  $\bar{{\mathscr{I}}}_0$  rewritten in a form that makes it easier to analyze the qualitative results of varying parameters. As $\epsilon$ approaches zero the basic invasion number reaches a minimum which will make it difficult for the ants to spread because the number of cells with pesticide will continue to increase and saturate the remaining cells. Biologically, this means that pesticide which is applied remains effective indefinitely. As $\epsilon \rightarrow \infty$ , $\bar{{\mathscr{I}}}_0$ now reaches a point where the spread of the ants is dependent on the von Neumann neighborhood and the likelihood that one of its neighbors will transmit a colony. The pesticide is now decomposing so quickly that it has no temporal significance in slowing down the ants. If we let $\mu \rightarrow 0$ the basic invasion number increases without bound and the ants will now effectively colonize unperturbed, meaning that there is no pesticide application at all. On the contrary, if we let $\mu \rightarrow \infty$ the ants will be sprayed so rapidly that their budding rate will play no part in their invasion success. 
The feasibility of any pesticide application regiment depends on the environmental tolerance for pesticide toxicity and the resources available.


\section{Numerical Results}
In this section we numerically approximate our system of differential equations for the MF model and the PA model. We used the sub-function $\it ode \,45$ in MATLAB 2013a . Then we plot the proportion of land with pesticide, the proportion of habitable land, and the proportion of land occupied by ants with respect to time. We compare the results from these two models with the CA model. For our first set of graphs, we generate initial conditions by randomly distributing initial states on a grid using a weighted probability. The odds of a cell being initialized as a 1 is 99\%, and the odds of a cell being initialized as a 2 is 1\%.

\begin{figure}[h!]
\centering
(a)\includegraphics[scale=0.35]{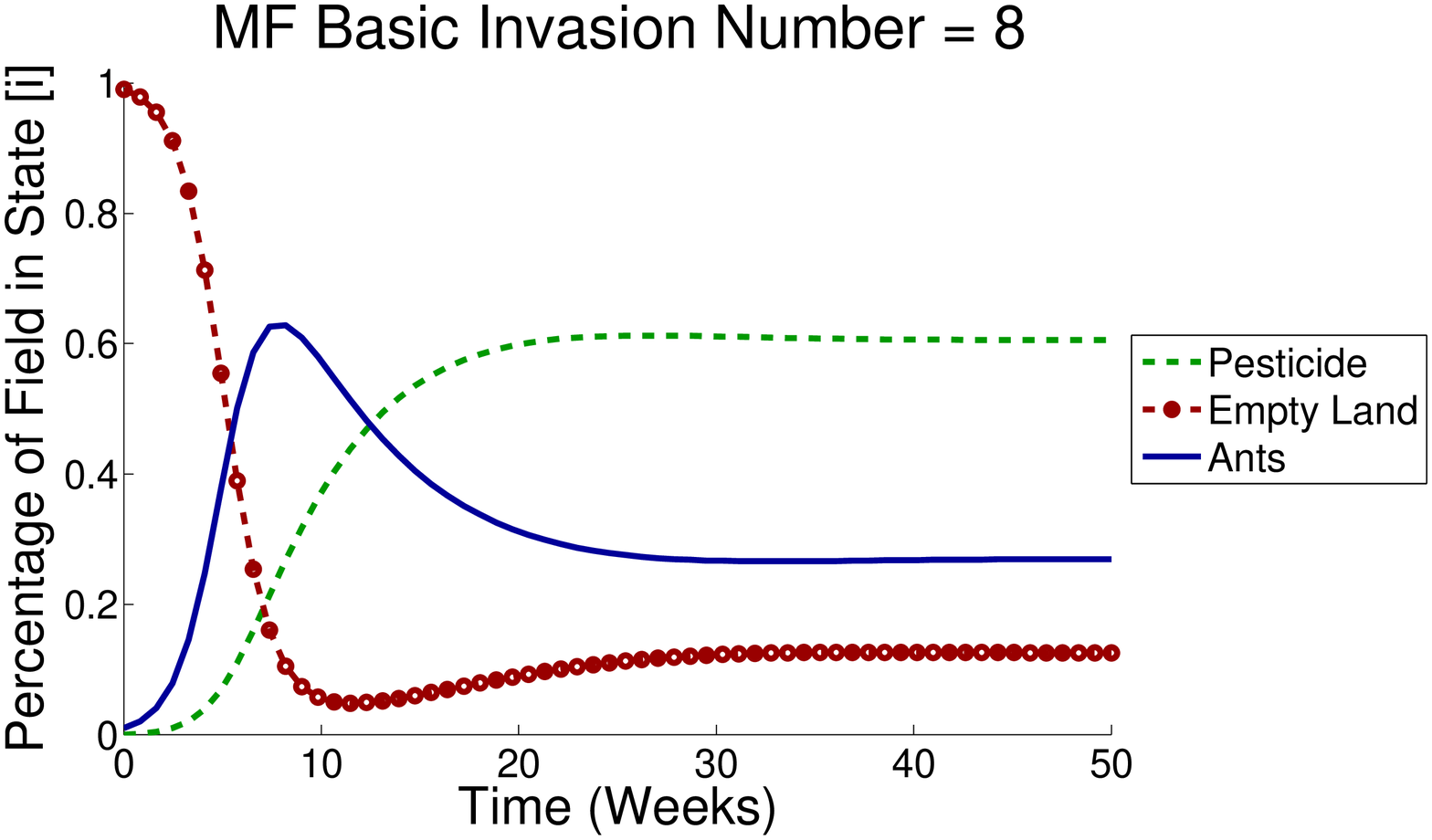}
\centering
(b)\includegraphics[scale=0.35]{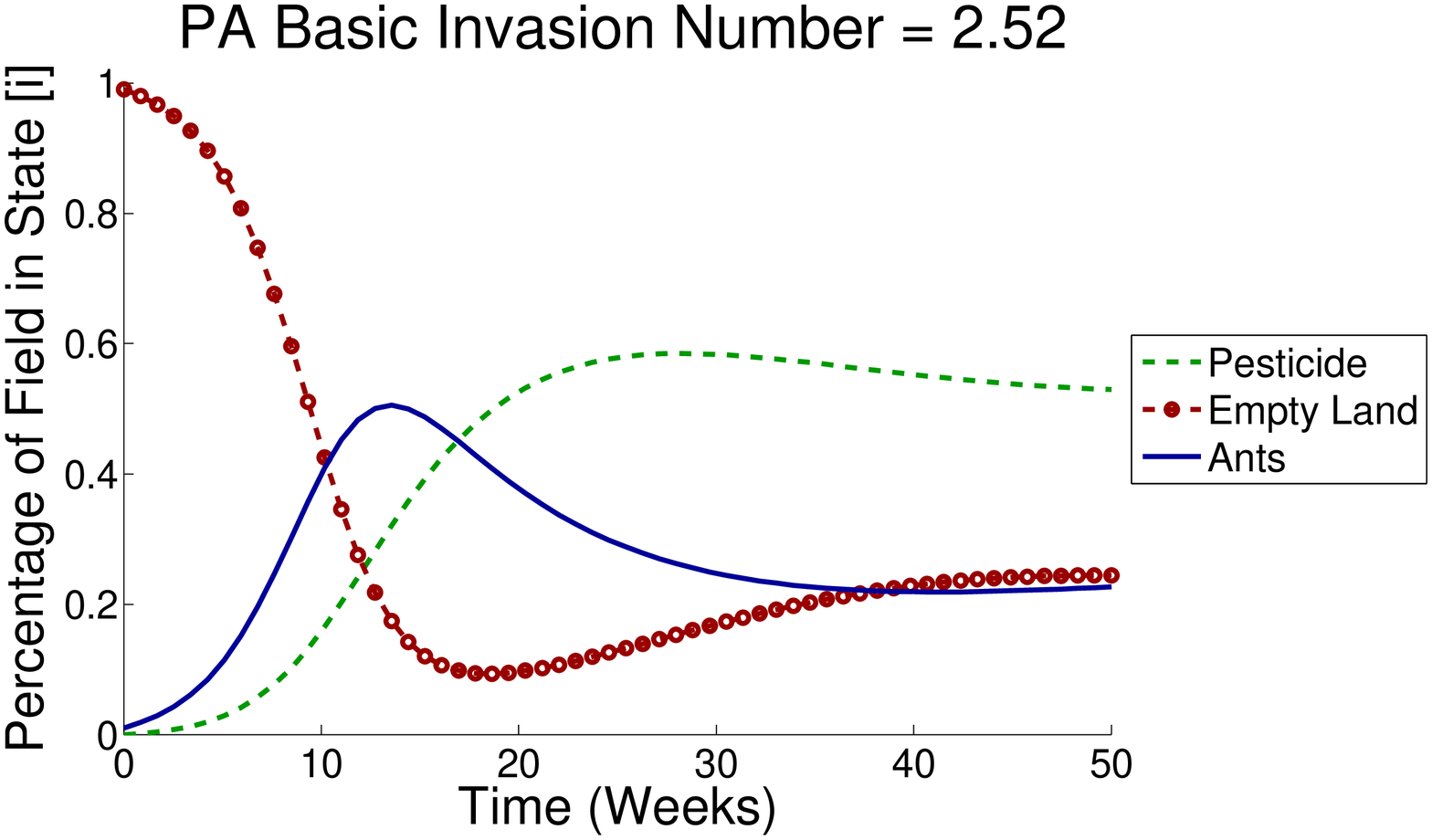}
\centering
(c)\includegraphics[scale=0.35]{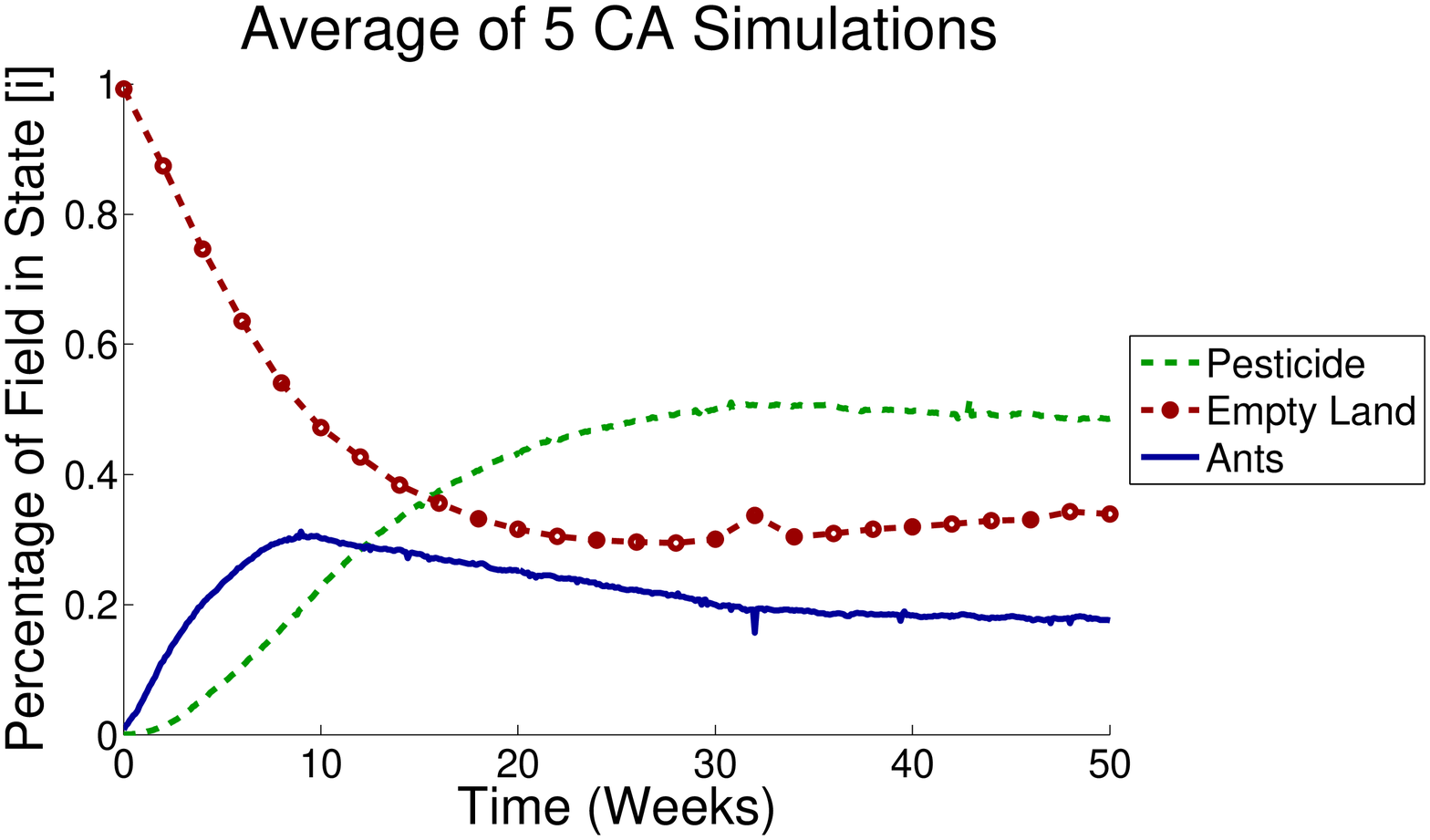}
\caption{\footnotesize (a) MF model solutions. (b) PA model solutions. (c) CA model solutions.\newline 
\centerline{$\phi=1, \mu = 1/8$, $\epsilon = 1/8$}}
\label{PlotAEE}
\end{figure}

From the graphs in Figure~\ref{PlotAEE}, the ants fare best in the MF model, followed by the PA model, followed by the CA model. The MF model assumes spatial independence of the cells; all cells are homogeneously mixed, and the ant colony survival depends on the relationship between $\mu$ and $\phi$. The PA model is also homogeneously mixed, but does not assume spatial independence.  The ants fare the worst in the CA model. In the CA model, explicit spatial dependence is taken into account. There are also possibilities of wasted birth, which cause the ants to fare the worst in the CA model. Figure~\ref{CASimu} shows a visualization of the simulation at the endemic equilibrium with the same initial conditions as the CA in Figure~\ref{PlotAEE}.
\begin{figure}[h!]
\centering
(a)\includegraphics[scale=0.35]{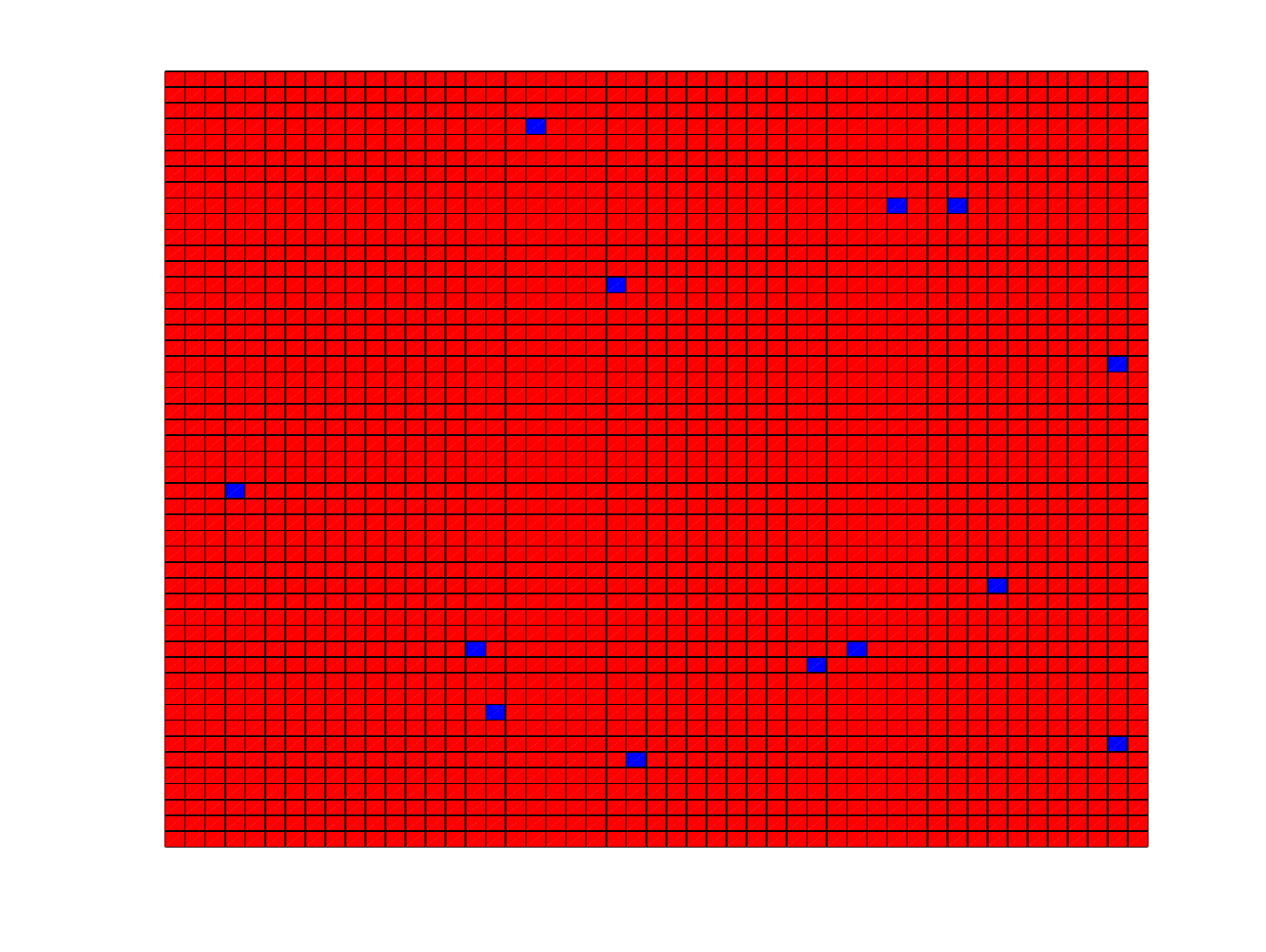}
\centering
(b)\includegraphics[scale=0.35]{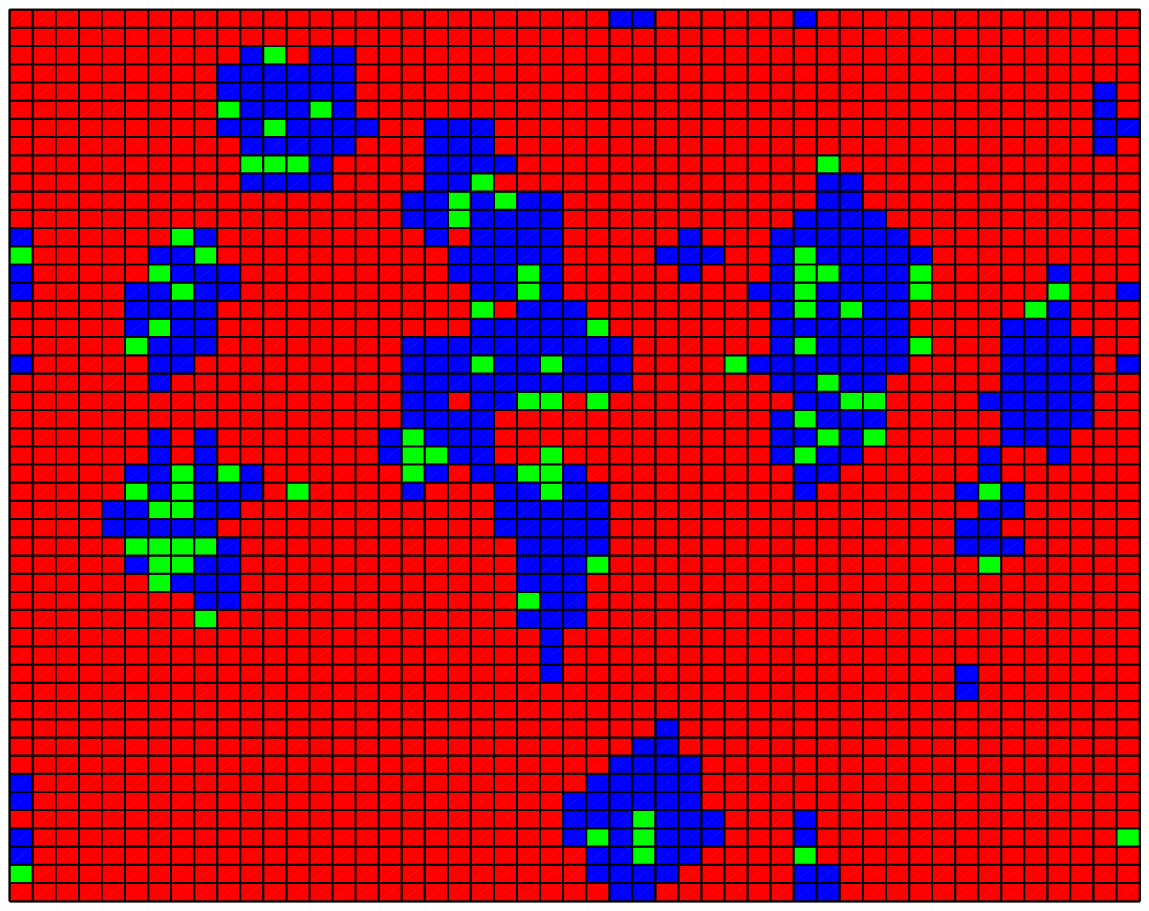}
\centering
(c)\includegraphics[scale=0.35]{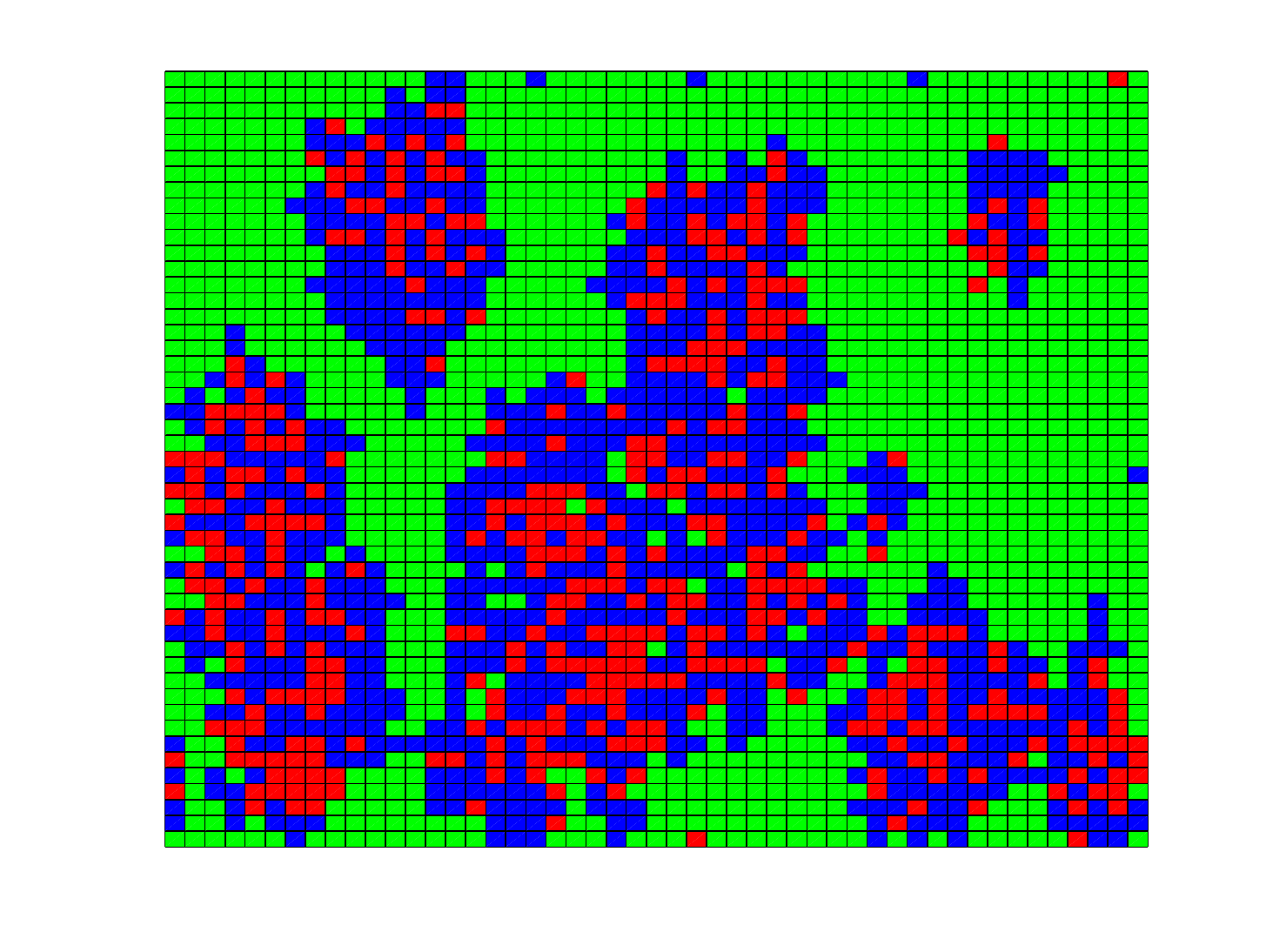}
\centering
(d)\includegraphics[scale=0.35]{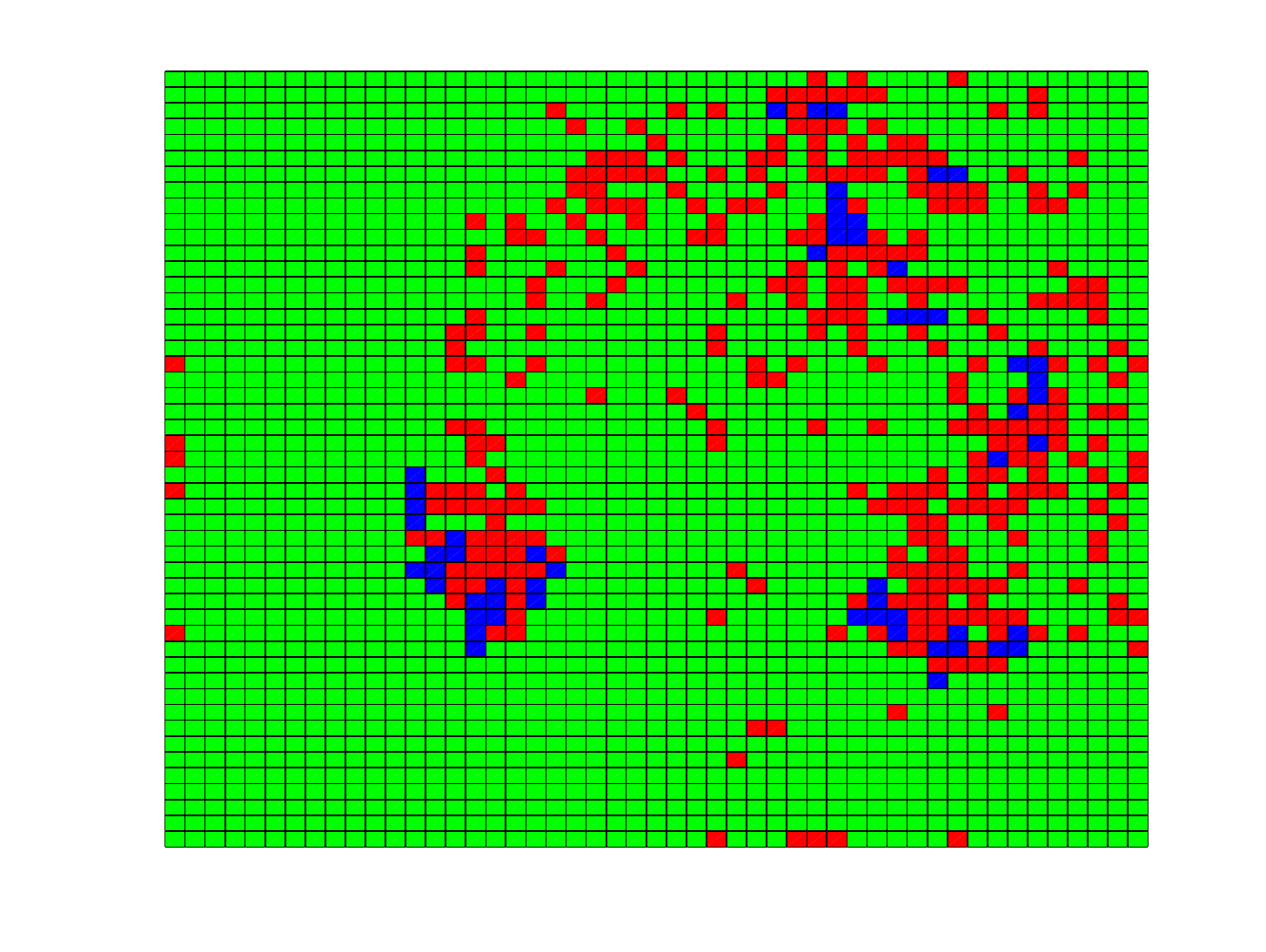}
\caption{\footnotesize CA Simulation/AEE Red - Free Land	Blue - Ant	Green - Pesticide	 (a) Initial Ant Invasion, t = 0 wks	(b) Intermediate ant invasion, t = 3.2 wks    (c) Peak ant invasion, t = 8.19 wks    (d) Final ant invasion, t = 88 wks}
\label{CASimu}
\end{figure}
For certain sets of parameters, our three models predict different long term behavior. The invasion threshold $\mathscr{I}_0$  from our MF model
predicts that the ant population will successfully colonize if $\phi > \mu$. However, the $\bar{\mathscr{I}}_0$ we received from our PA model predicts that when $\phi < \frac{4\mu(\mu+\epsilon)}{2\mu + 3\epsilon}$, our ants will not survive. There are some sets
of parameters where both of these conditions can be satisfied: For $ 2\mu >\phi > \frac{4\mu}{3}$, and 
$\epsilon < \frac{2\phi\mu - 4\mu^2}{4\mu - 3\epsilon}$, both of the above conditions are satisfied and the MF model disagrees
with the PA model. In Figure~\ref{PlotAEE}, we plotted a set of graphs from the three models with parameters that satisfy these conditions to determine 
the long term behavior of our ant colonies.

\begin{figure}[h!]
\centering
(a)\includegraphics[scale=0.35]{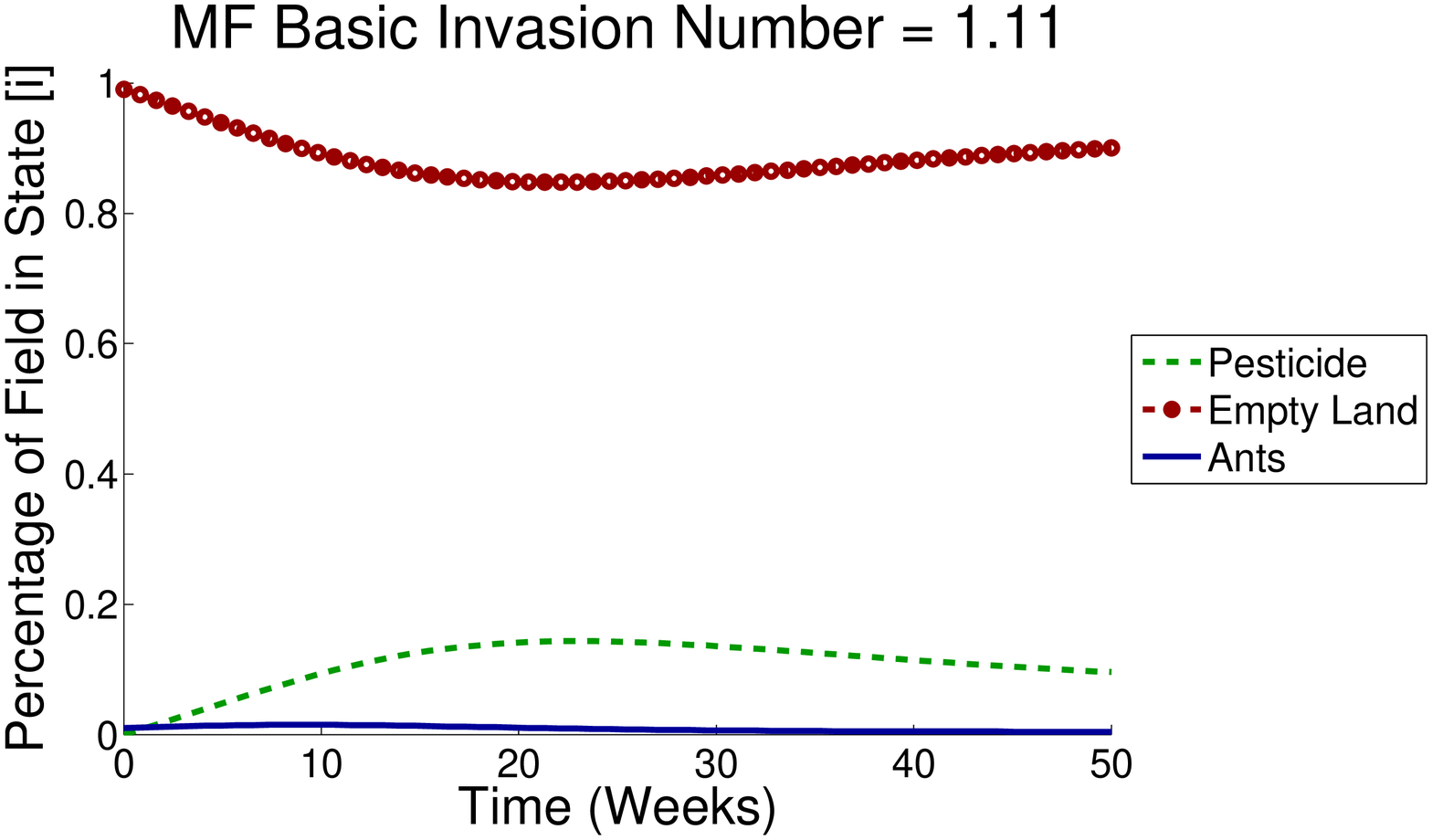}
\centering
(b)\includegraphics[scale=0.35]{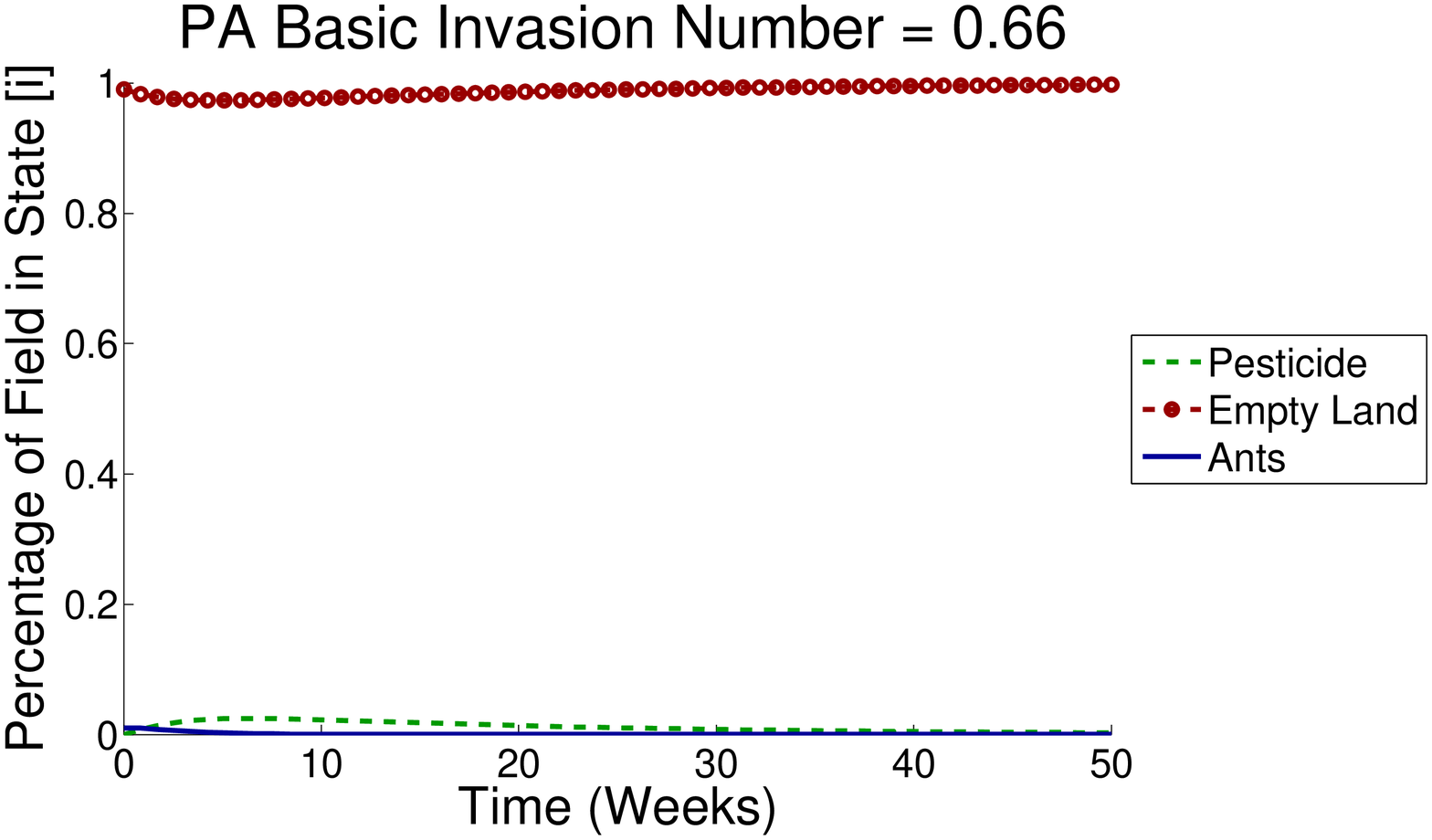}
\centering
(c)\includegraphics[scale=0.35]{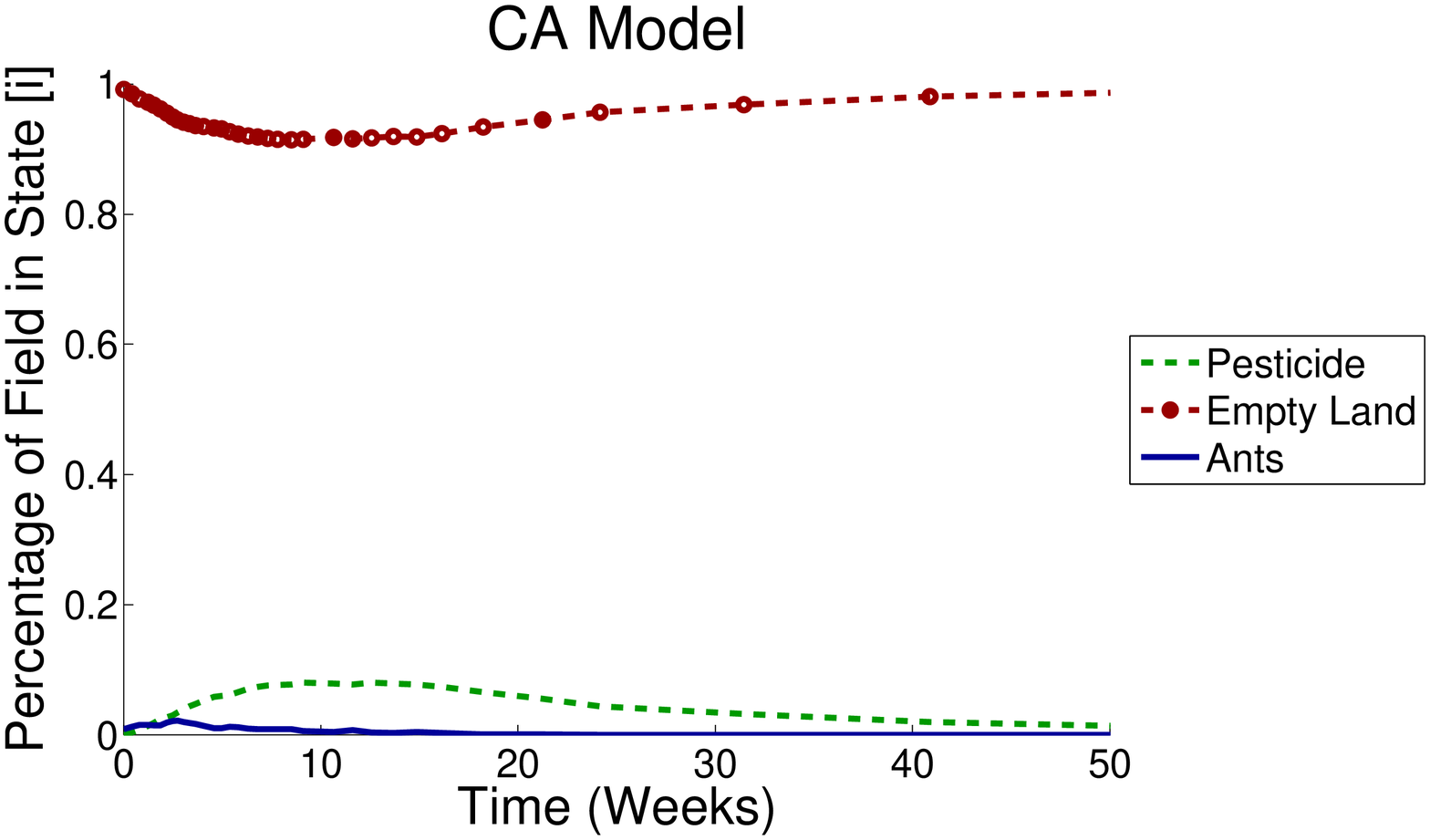}
\caption{\footnotesize (a) MF model. $\mathscr{I}_0 = 1.11$(b) PA model. $\bar{\mathscr{I}}_0 = .66$(c) CA model.}
\label{PhiSweeps}
\end{figure}

Although the MF model predicts the survival of the ants, the predicted ant colony density is so low that it effectively agrees with
the PA and CA models. However, because the ants do not perish in the MF model, the density of cells affected by pesticide do not
approach zero in the MF model. This difference between the MF model and the other two, suggests that for
certain sets of parameters the MF model does not accurately predict the density of pesticide application. If $\phi>\frac{4\mu}{3}$ then $\phi>\mu$ and the ants survive in the MF Model. If $\phi<\frac{4\mu(\mu+\epsilon)}{2\mu+3\epsilon}$ then $3\phi\frac{(\mu  +\epsilon)}{\mu (\phi + 4 \mu + 4\epsilon)}<1$ and the ants die out in the pair approximation model.
So if $\frac{4\mu(\mu +\epsilon)}{2\mu+3\epsilon}> \phi>\frac{4\mu}{3}$ the ants survive in the MF field model, die out in the PA model, and are more likely to die out in the CA. This is indicative of the large-scale, spatial independence of the MF model, in which ants may colonize indiscriminately on any empty lands. This model produces an ``overall picture" of how effectively the ants are spreading but may overstate their success in overcoming specific obstacles. The locality of impediments impediments to colony budding, which are the saturation of surrounding land with colonies or pesticide of a colony ready to bud, in the PA model makes it more difficult for the ants to overcome relatively high rates of pesticide application. The complete incorporation of all local interactions in the CA simulation means that spread of the ants is the least likely in situations where pesticide application is close to the budding rate.  This aggregate account of interactions in CA model seems to be more realistic for parameters as described above.
\subsection{Parameter Sweeps}
\begin{figure}[h!]
\centering
(a)\includegraphics[scale=0.35]{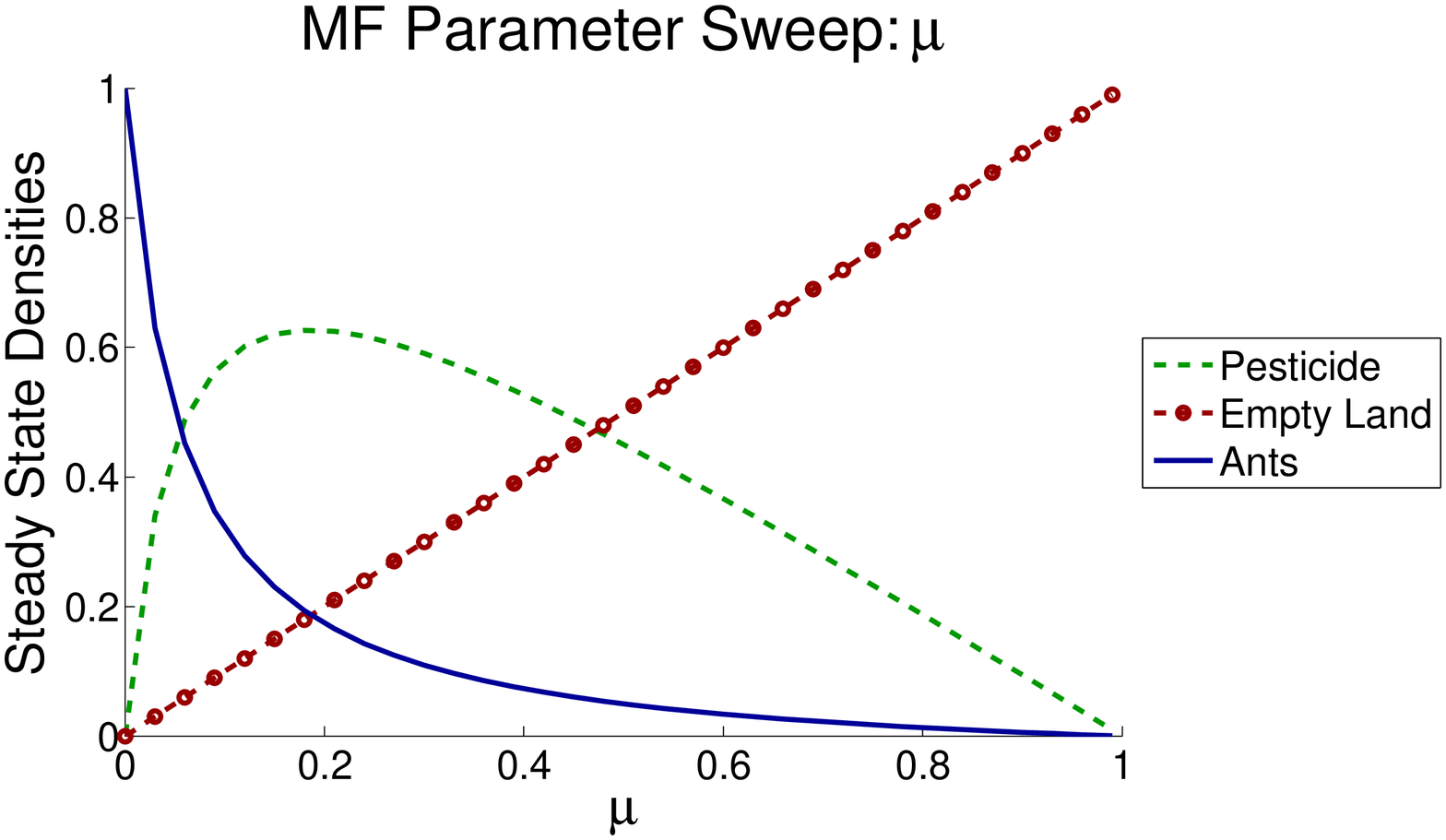}
\end{figure}
\begin{figure}[h!]
\centering
(b)\includegraphics[scale=0.35]{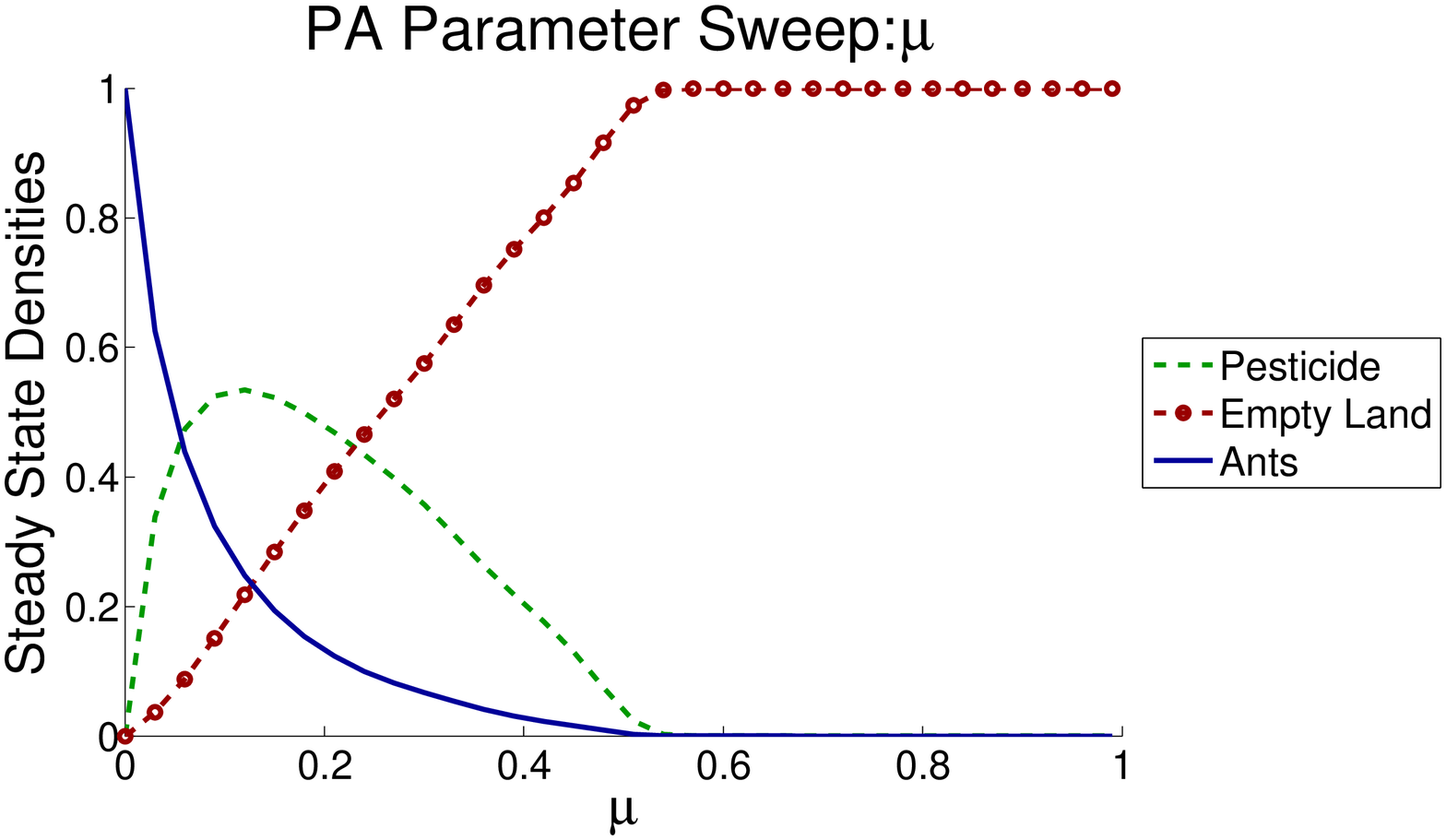}
\end{figure}
\begin{figure}[h!]
\centering
(c)\includegraphics[scale=0.35]{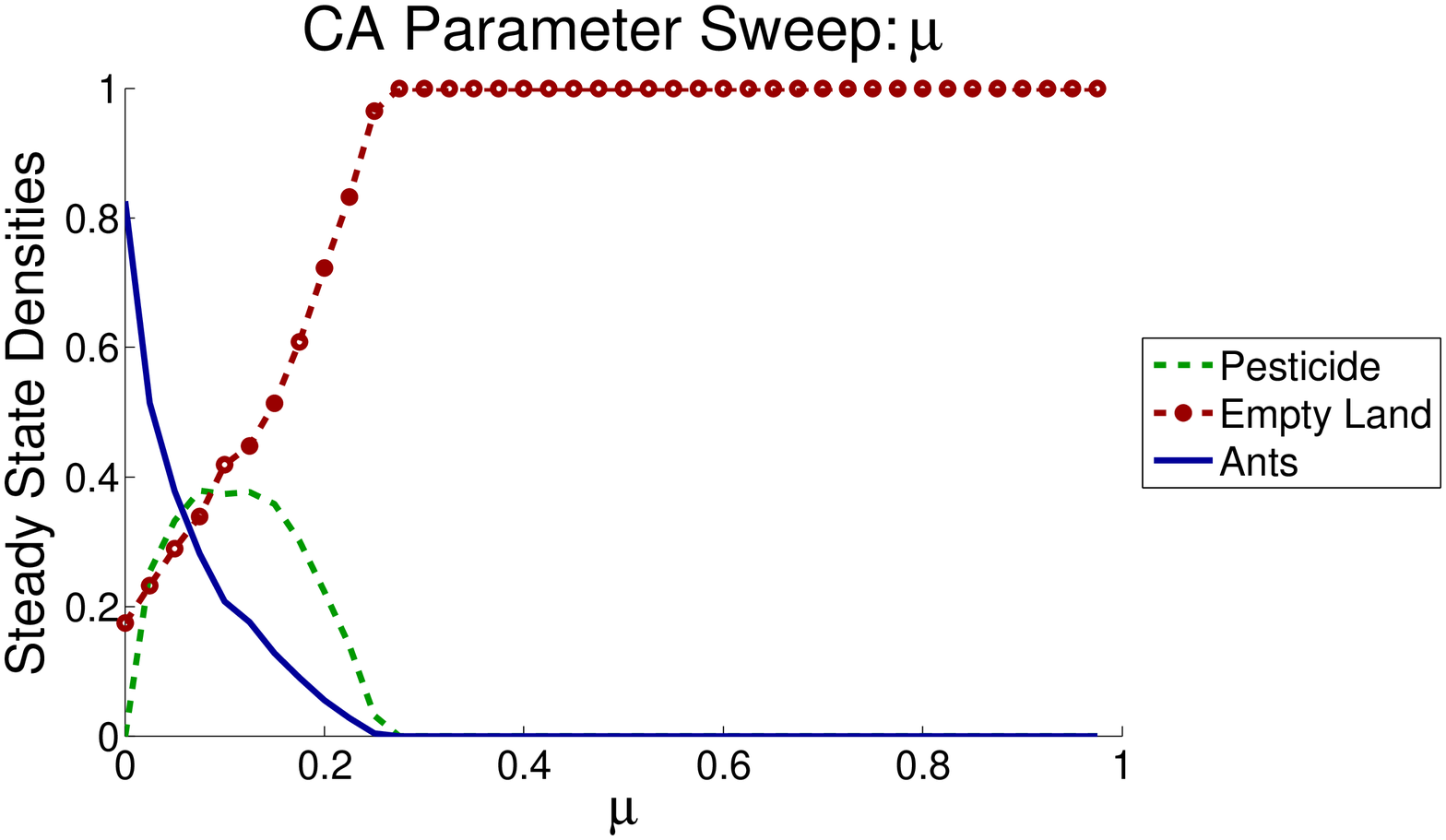}
\caption{\footnotesize (a) MF $\mu$ Parameter Sweep (b) PA  $\mu$ Parameter Sweep (c) CA  $\mu$ Parameter Sweep}
\label{MuParamSweep}
\end{figure}
We conducted a parameter sweep with respect to $\mu$, where $\phi$ and $\epsilon$ are constant. We chose $\epsilon = 1/18$ and
$\phi = 1$. The results are presented in Figure \ref{MuParamSweep}. As we increase $\mu$, the proportion of the ant colonies decreases. Biologically, as the pesticide application rate increases, ant colonies die faster. Once $\mu$ reaches the value where the maximum pesticide density is attained, the ants cannot colonize fast enough, and they will die out. Consequently, the proportion of free land should increase as $\mu$ increases.

The steady-state behavior of the proportion of pesticide treated land is not strictly decreasing or increasing with respect to $\mu$.
As $\mu$ increases, so does the steady-state proportion of pesticide treated land. However, there is an approximate maximum value of 0.6 for the steady-state proportion of pesticide treated land. As $\mu$ continues to increase,  the steady-state proportion of pesticide treated land decreases. At low $\mu$, pesticide distribution logically creates more land with pesticide. However, there is a critical value of $\mu$ where the pesticide application kills off ants so efficiently that less and less regions need to be sprayed with pesticide. When $\mu=0.3$, the ants completely die out, and no more pesticide needs to be applied. Biologically, this means that slow pesticide application curtails the population of ant colonies, but leaves them at a high proportion. Such an approach takes into account safety concerns but does not eliminate the biological invasion. Increasing pesticide application controls the ant colonies at lower proportions but requires higher levels of pesticide. At a certain point, aggressive application pays off  because the ant colonies are left at a very low steady state and pesticide levels can be lowered to safer levels. In summary, aggressive application quickly eliminates the ant infestation and leaves them at a low steady state, but at the cost of increased health and environmental concerns.

In the second parameter sweep, we varied the colony budding rate $\phi$ on the interval [0,2] for the MF, PA, and CA. We held $\mu=0.5$ and $\epsilon=\frac{1}{18}$. At low colony budding rates the ants are quickly exterminated, and empty and pesticide lands do not experience changes. In the MF model at $\phi=0.5$, there is a qualitative change in $P[2]$ which is to be expected since ${\mathscr{I}}_0 =1$, the basic invasion number, becomes greater than one. The ant colonies survive at low levels and the proportion of pesticide sites grows quickly for small increases in $\phi$, while the proportion of empty sites drops sharply. Such steady state transitions are expected since the AEE begins to grow while $\phi >\mu$. The PA model displays similar qualitative behavior, although the ant colonies become endemic at a higher value of $ \phi \approx .9 $ since incorporation of spatial distribution prevents high contact rates between empty land sites and ant colony sites. For both the MF and PA models the increase in $\phi$ produces a greater steady-state endemic point for the ants, but there is an asymptotic limit as $\phi \gg 1$. This limit for the endemic equilibrium seems to be close to  0.1. Indeed, if we refer back to the AEE$=\left( \frac{\mu}{\phi},\frac{\epsilon}{\mu+\epsilon}(1-\frac{\mu}{\phi}) \right)$,
\begin{equation}
 \lim_{\phi \rightarrow \infty}
 \mbox{AEE} = \left( 0,\frac{\epsilon}{\mu+\epsilon} \right).
\end{equation}
\indent It can be seen from both the MF and PA models that $\frac{\epsilon}{\mu+\epsilon}$ qualitatively seems to be an upper bound for the proportion of ant colony sites $P[2]$. However, empty sites $P[1]$, experience sharp decreases until there is a saturation of ant colonies and pesticide sites. We can also see that for high rates of pesticide degradation, the ant endemic upper bound can be increased. Biologically, this means that if an application of pesticide is short-lived, ant colonies will quickly reoccupy the land after the initial colony is eliminated.  Such a pesticide would be highly ineffective and provide less benefits in relation to its possible harmful effects. On the contrary, low rates of $\epsilon$ prevent quick recolonization and are in fact a sort of limiting factor to ant colony budding. If we rewrite the results from {25} as \[\frac{1}{\frac{\mu}{\epsilon}+1}\] we can see the relationship between $\epsilon$ and $\mu$. In fact, even at very high colony budding rates, the ant invasion can be terminated if the pesticide remains effective for a long time. The CA simulation shows similar behavior to the PA model, with  the proportion of ant colony sites $P[2]$ remaining at low levels despite increasing $\phi$. These numerical and analytic results suggest that using pesticide that decomposes slowly may be more effective than aggressive spraying in controlling initial ant invasions.

\begin{figure}[h!]
\centering
(a)\includegraphics[scale=0.35]{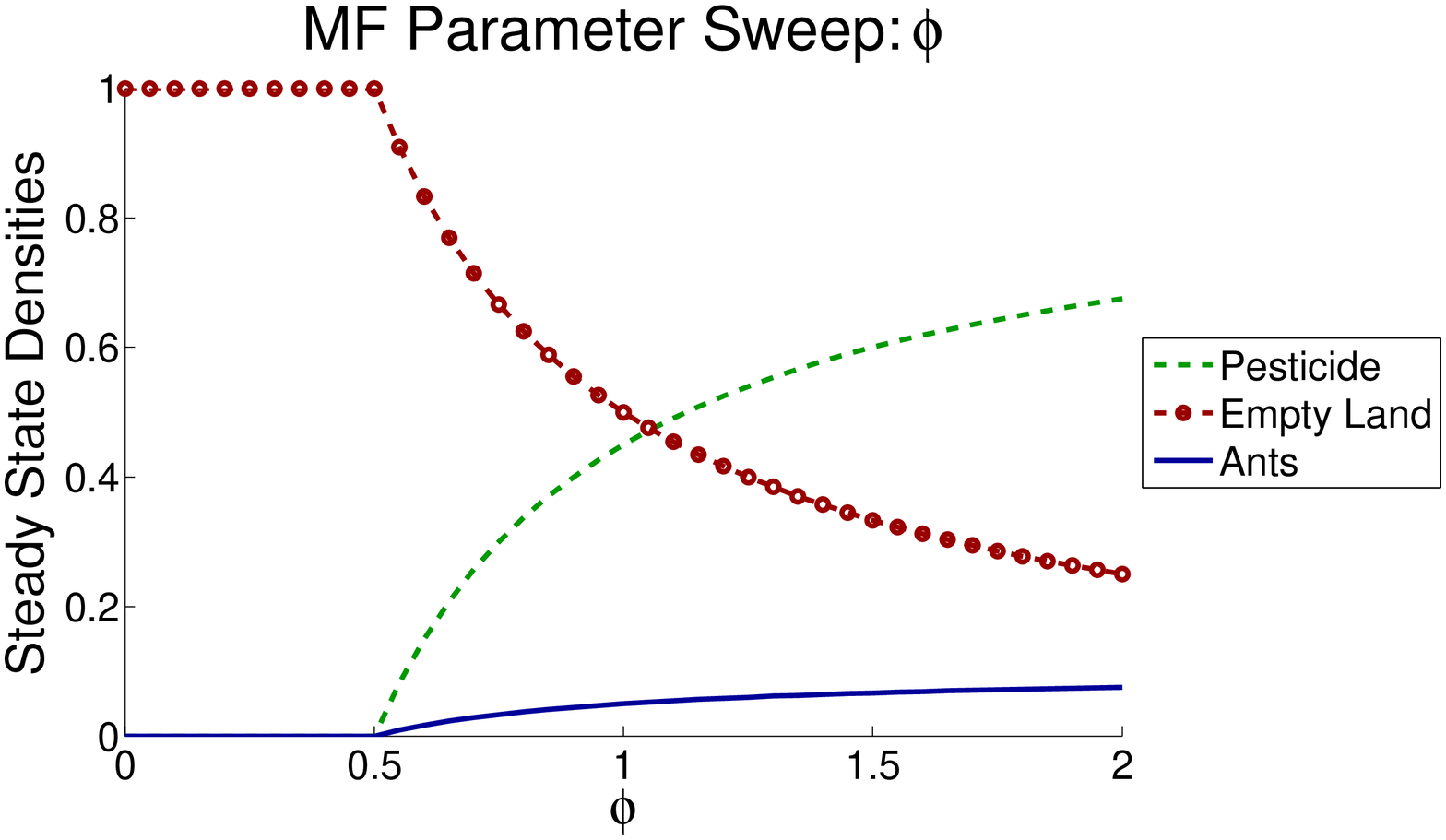}
\centering
(b)\includegraphics[scale=0.35]{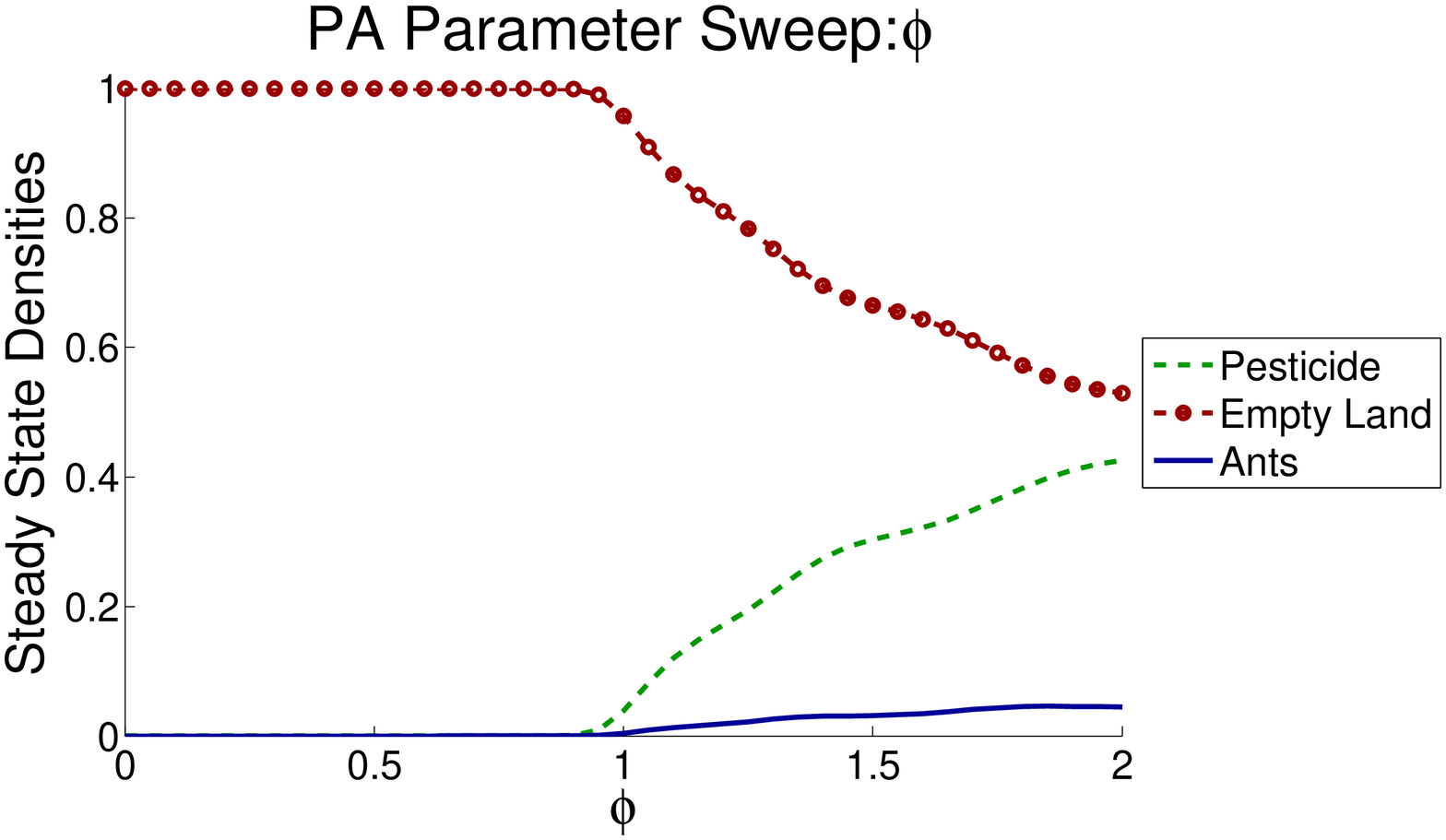}
\centering
(c)\includegraphics[scale=0.35]{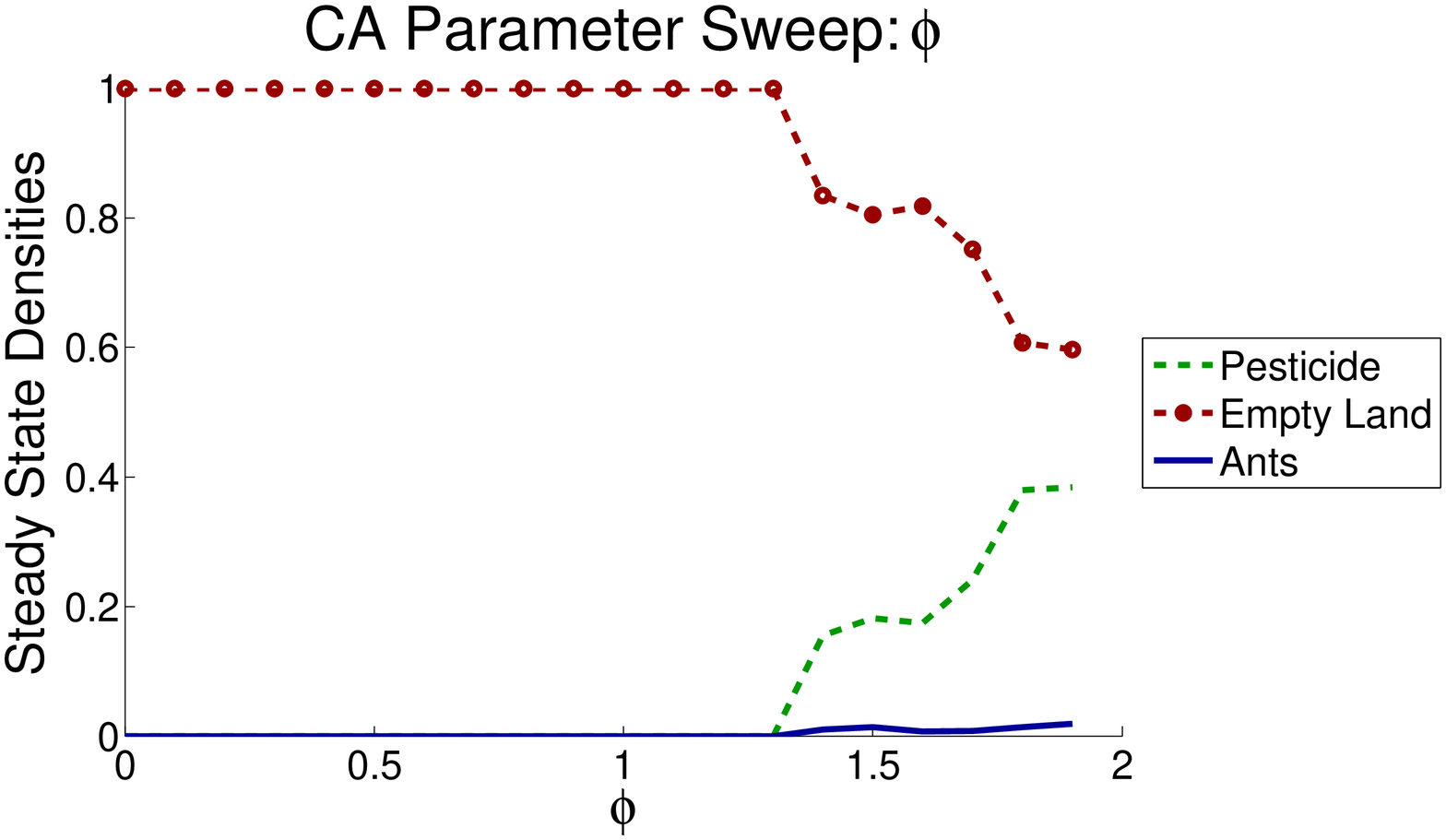}
\caption{\footnotesize (a) MF $\phi$ Parameter Sweep (b) PA $\phi$  Parameter Sweep (c) CA $\phi$ Parameter Sweep}
\end{figure}

\section{Discussion}
Our first research goal was to develop an understanding of the current Rasberry ant biological invasion in the U.S. Gulf states.
Our  second  goal was to construct a probabilistic spatiotemporal simulation that models the spread and control of an invasive ant species, considering both  pesticide application and degradation. We utilized a CA framework to develop a 2-D lattice with three states and three events defined according to the biological assumptions of the ant colony budding process and pesticide application and degradation. The third goal was to build two deterministic models that approximate our spatially explicit CA simulation with the same biological assumptions.

The MF omits spatial relationships and assumes homogeneous mixing of states. Ant colony budding success is solely dependent on the presence of empty land in the  lattice, by passing the requirement that such a process can only occur if there is adjacent empty land. According to our analytic results of this model, the invasion success of the ant colonies only depends on the ratio $\frac{\phi}{\mu}$. While this reflects the essential relationship between the rate at which ant colonies grow and   the rate at which they are suppressed with pesticide, it overstates the possibility of local empty land being readily available. For example, an ant colony which is surrounded by pesticide sites is free to colonize in any empty sites regardless of its location. The result is that this model overstates the likelihood that the ants will survive pesticide application and remain endemic. Relative to the other models, the MF predicts a higher and earlier ant invasion peak, which emphasizes the need for pesticide application very early on.

The PA model includes the implicit spatiotemporal local interactions between states and more realistically describes the spread of ants in a 2D lattice. The different combinations of adjacent state pairs in the PA ODEs effectively take into account the probability of ant colonies having local availability of empty land. The numerical results show that the ants reach extinction faster than in the the MF model, and survive at lower proportions for parameter conditions that produce endemic equilibria. This result was predicted from the derivation of $\mathscr{I}_0$ and $\bar{\mathscr{I}}_0$:
\[\bar{\mathscr{I}}_0 = \mathscr{I}_0 \frac{3(\mu+\epsilon)}{\phi + 4\mu + 4\epsilon} < \mathscr{I}_0,\] where the MF invasion threshold number $\mathscr{I}_0$ is independent of the rate of degradation and $\bar{\mathscr{I}}_0$ depends on $\epsilon$. In the PA model, pesticide application not only eliminates an ant colony, but a cell sprayed with pesticide cannot be colonized until the pesticide degrades. This suggests there may be more time to respond to the ant invasion than the MF predicts, and apply pesticide at lower rates. In terms of public health and environmental concerns, the PA model shows that there are less detrimental strategies to control the ant invasion. 

	The CA simulation used a 2D lattice and Poisson event process to model the probabilistic spread of ants.  The bookkeeping of state changes throughout space and time in our CA simulation can be interpreted as generating data about the spread of Rasberry ants {\it  in silico}. The simulations showed that for colony budding rates less than the  pesticide application rate, the invasion died out fairly quickly. This suggests that if the ants do not have good environmental conditions, their spread may be stopped with minimal pesticide application.  For colony budding rates less than twice the pesticide application rate, the simulations showed that the invasions did not always become extinct and the pesticide degradation rate played a bigger role in determining the steady state. Consequently, even at low colony budding rates and high pesticide application rates, the ants may linger and the success of intervention depends on the pesticide remaining effective for longer periods of time.  This presents a conflict of interest between eradicating the ant invasion and prolonged pesticide contamination in the environment. For budding rates twice the pesticide application rate the simulations showed that the ant invasion became endemic. Rapid ant colony growth may require increased pesticide dispersal at the cost of health and environmental concerns.  A shortcoming of our CA simulation was the inaccessibility to faster processors to average more trials and lessen the stochastic anomalies. In the future, we would like to improve our code in order to strengthen the conclusions of the CA simulations.

\section{Future Work}
The models and simulation presented in this paper are limited. The colonization process is not colony size dependent and it is assumed that a small colony colonizes as often as a larger colony.  Another assumption made in the presented models was that the ants could only colonize in the cardinal directions and all land has equal amount of resources. The pesticide is also assumed to be 100\% effective though most pesticides are not 100\% effective and experimenting with different kill rates could prove fruitful. Further work would include a colony dependent birth rate. This would increase the validity of our models because in nature, the budding process is dependent on the colony size and resources. Also we would assume the ant colonies can move within a Moore neighborhood ( four cardinal directions including Northeast, Northwest, Southeast, and Southwest). This will increase the dynamics in potential colonization. If we assumed amount of resources is cells, we could get insight on real life probabilistic colonization rates and implement strategies to drive invasive ant species to extinction.

\section*{Acknowledgments} 
We would like to thank Dr.~Carlos Castillo-Chavez, Executive Director of the Mathematical and Theoretical Biology Institute (MTBI), for giving us the opportunity to participate in/for specialized lectures.  We would also like to thank Co-Executive Summer Directors Dr.~Erika T.~Camacho and Dr.~Stephen Wirkus for their efforts in planning and executing the day to day activities of MTBI. We would also like to extend gratitude to Kamal Barley for his graphic enhancements and to Ben Morin for the personalized lectures. We would also like to extend gratitude to Marta Sarzynska for revisions and illustrations,  to Dustin Padilla for his mathematical revisions, and to Dr Baojun Song for his pearls of wisdom. This research was conducted in MTBI at the Mathematical, Computational and Modeling Sciences Center (MCMSC) at Arizona State University (ASU). This project has been partially supported by grants from the National Science Foundation (NSF - Grant DMPS-1263374), the National Security Agency (NSA - Grant H98230-13-1-0261), the Office of the President of ASU, and the Office of the Provost of ASU.

          \bibliographystyle{plain}
          \bibliography{newantbib}

\newpage
\section{Appendix}
\subsection{Next Generation Operator}
X is the vector function which contains all pair sites with ant colonies, and we will call them infested sites.  $\dot{X}=$ is the rate of change of infested sites, which can be rewritten as\[ \dot{X}=\mathscr{F}-\mathscr{V}\]
The vector $\mathscr{F}$ is the  birth of new colony sites coming from other colony sites, and $\mathscr{V}$ is changes in ant colony sites from the progression of the invasion(i.e. pesticide application).
\begin{eqnarray*}
\dot{X}& = & \begin{bmatrix}
P[02]\\
\\
P[12]\\
\\
P[22]
\end{bmatrix} \\ \\ \\
\mathscr{F} &=&\begin{bmatrix}
& \frac{3}{4}P[01] \phi Q_{2|1} & \\
\\
& P[11] \left(\frac{3}{4}\phi Q_{2|1} \right) & \\
\\
& 0 &
\end{bmatrix}, \; \;
\mathscr{V} = \begin{bmatrix}
-\mu P[22] + P[02] \left(\mu + \epsilon \right)\\
\\
-\epsilon P[02] + P[12] \left(\frac{3}{4} \phi Q_{2|1}+\frac{\phi}{4} + \mu \right) \\
\\
-2P[12] \left(\frac{3}{4} \phi Q_{2|1} + \frac{\phi}{4} \right) + 2P[22] \mu
\end{bmatrix}
\end{eqnarray*}
We take the Jacobian of $\mathscr{F}$ and $\mathscr{V}$. Let $F=Jacobian(\mathscr{F} )$ and $V=Jacobian(\mathscr{V})$
\begin{equation*}
F=\begin{bmatrix}
0 & \frac{3}{4} \phi P[13] \frac{ \left(P[01] + P[11] \right)}{ \left( P[01] + P[11] + P[12] \right)^2} & 0 \\
\\
0 & \frac{3}{4} \phi P[11] \frac{ \left(P[01] + P[11] \right)}{ \left( P[01] + P[11] + P[12] \right)^2} & 0 \\
\\
0 &  0 & 0
\end{bmatrix}
\end{equation*}
\begin{equation*}
V=\begin{bmatrix}
\left( \mu + \epsilon \right) & 0 & -\mu \\
\\
-\epsilon & \frac{3}{4} \phi \left( \frac{P[12]}{P[01]+P[11]+P[12]} + \frac{P[12] \left(P[01] +P[11] \right)}{ \left( P[01] + P[11] + P[12] \right)^2} \right) + \frac{\phi}{4} +\mu & 0 \\
\\
0 &  -2 \frac{3}{4} \phi \left( \frac{P[12]}{P[01]+P[11]+P[12]} + \frac{P[12] \left(P[01] +P[11] \right)}{ \left( P[01] + P[11] + P[12] \right)^2} \right) - 2\frac{\phi}{4}  & 2 \mu
\end{bmatrix}
\end{equation*}
\subsection{Pair Approximation AFE}
\begin{figure}[h!]
\begin{center}
\caption{\footnotesize Symmetric Proof}
\includegraphics[scale=.75]{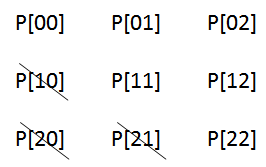}
\label{Symmetric Proof}
\end{center}
\end{figure}

\begin{figure}[h!]
\begin{center}
\caption{\footnotesize dP[01] Description}
\includegraphics[scale=.75]{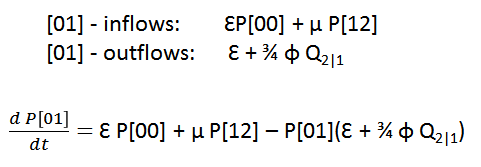}
\label{dP}
\end{center}
\end{figure}

\[(P[00],P[01],P[02],P[11],P[12],P[22])\]

\[AFE = (0,0,0,1,0,0)\]

\begin{equation*}
F(AFE)=\begin{bmatrix}
0 & 0 & 0
\\
0 & 3 \frac{\phi}{4} & 0
\\
0 & 0 & 0
\end{bmatrix}
\end{equation*}

\begin{equation*}
V(AFE)=\begin{bmatrix}
\mu + \epsilon & 0 & -\mu
\\
-\epsilon  &  \frac{\phi}{4} + \mu & 0\\
\\
0 & -\frac{\phi}{2} &  2\mu
\end{bmatrix}
\end{equation*}

\begin{equation*}
FV^{-1}=\begin{bmatrix}
 0 & 0  & 0 \\
\\
 \frac{3 \phi \epsilon}{\mu (\phi + 4 \mu +4\epsilon)} & 3\phi\frac{(\mu  +\epsilon)}{\mu (\phi + 4 \mu + 4\epsilon)} &\frac{3}{2} \frac{\phi \epsilon}{\mu (\phi + 4 \mu+ 4\epsilon)} \\
\\
 0 & 0 & 0
\end{bmatrix}
\end{equation*}
The Eigenvalues of $FV^{-1}$
\begin{equation*}
\begin{bmatrix}
 & 0 & \\
\\
& 0 & \\
\\
& 3\phi\frac{(\mu  +\epsilon)}{\mu (\phi + 4 \mu + 4\epsilon)} & 
\end{bmatrix}
\end{equation*}
Basic Invasion Number $\mathscr{I}_0=3\phi\frac{(\mu  +\epsilon)}{\mu (\phi + 4 \mu + 4\epsilon)}$
\\
\subsection{Original Flow Chart}
The following chart describes the changes in states in order according to the events possible at each site.

\begin{figure}[h!]
\begin{center}
\caption{\footnotesize Flow Chart}
\includegraphics[scale=.75]{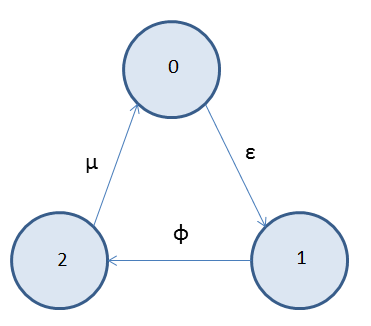}
\label{Flow Chart}
\end{center}
\end{figure}

\subsection{Stochastic Simulation}
Our Cellular Automaton model is a stochastic process which depends on a Poisson process for the occurrence of events. The realization of the events is  dependent on the presence of sites needed for a particular event to occur. However, the ant colony birth event is also dependent on the density and spatial location of ant colony sites $[1]$ and empty lands $[2]$, since they must occur adjacently. An ant colony must be next to an empty site to be a potential birth location. This means that a birth event can be wasted if a cell in state $[2]$ does not have any neighbors in state $[1]$. In figure 12 there is an example of a birth event that has been selected and a random cell in state $[2]$ has been chosen. It has no neighbors in state $[1]$ , which means the event will be a wasted event.
\begin{figure}[h!]
\begin{center}
\caption{\footnotesize Wasted Birth Diagram}
\includegraphics[scale=.70]{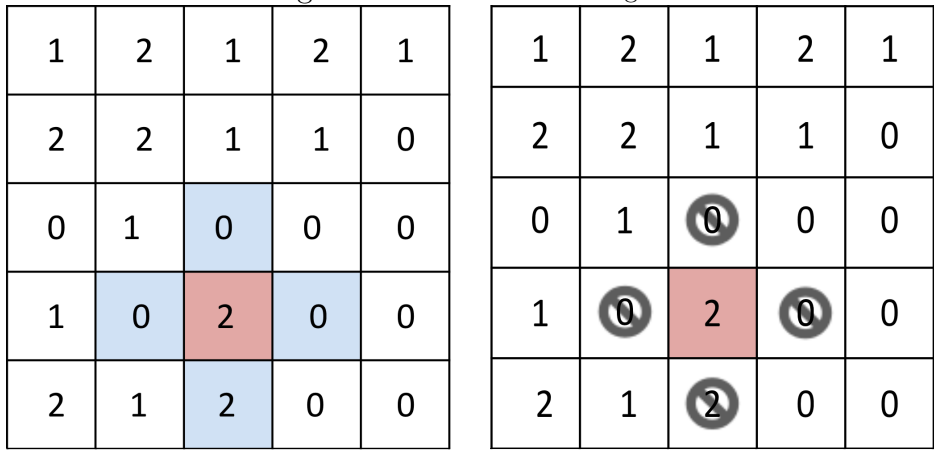}
\label{Wasted Birth}

\end{center}
\end{figure}

\subsection{Stochastic Simulation Code}
Below is the code we used to create our stochastic simulation.
\begin{verbatim}
#############################################################################
#############################################################################
#############################################################################
## Authors: Valerie Cheathon, Agustin Flores, Victor Suriel, Octavius Talbot,
##          Adrian Smith
##
## Stochastic Simulation of Ant Colony Spread with Pesticide application
##
########################################################################
##
##This code will set up a cell lattice which represents the spatial location
##of ants, habitable spots without ants, and pesticide application.
##
##In the matrix, a 0 entry corresponds to pesticide, 1 corresponds to land without
##ants, and 2 to land with ant colony present
##
########################################################################

#---------------------------------------------------------------------------------
#---------------------------------------------------------------------------------
#---------------------------------------------------------------------------------
# Parameters and Initial Condition
#----------------------------------------------------------------------------------
#----------------------------------------------------------------------------------
#----------------------------------------------------------------------------------
set.seed(100)
n<-50                    # The size of the n x n matrix to represent the cells.
A<-matrix(ncol=n,nrow=n) # Generating the null matrix
S0_0<-0                 # initial percentage of land with pesticide
S1_0<-0.99                # initial percentage of habitable land without ant colony
S2_0<-.01                 # initial percentage of land with ant colony
S21<-0                   # initial number of ant/empty cell pairs

S0<-NULL         # land sprayed with pesticide at ith time
S1<-NULL         # habitable land without ants at ith time
s2<-NULL         # land with ant colony at ith time

S0t<-0        # land sprayed with pesticide throughout time
S1t<-2479       # empty land throughout time
S2t<-21
# land with ant colony throughout time

#-------------------------------------------------------
# event rates: 3 possible events
#-------------------------------------------------------
mu<-1/8  # pesticide spray rate
epsilon<-1/18   # pesticide decomposition rate
phi<-1     # colony budding rate

P0<-NULL
P1<-NULL
P2<-NULL

#---------------------------------------
# initialize cumlative probabilities
#---------------------------------------
cump0<-NULL
cump1<-NULL
cump2<-NULL

tr<-NULL     # initialize total rate of all events
v<-NULL
time<-0      # initial time step

#------------------------------------------------------
#------------------------------------------------------
# This generates the n X n matrix with entries 0,1,2
#------------------------------------------------------
#------------------------------------------------------
for (i in 1:n){
  for (j in 1:n){
    x<-runif(1,0,1)
    if (x < S0_0){
      A[i,j]<-0
    } else if (x >= (S0_0) && x<=(S0_0+S1_0)){
      A[i,j]<-1
    } else if ( x >= (S0_0+S1_0)){
      A[i,j]<-2
    }

  }  # end inside forloop

} # end outside forloop

#-------------------------------------------------------------------------------------
#-------------------------------------------------------------------------------------
#-------------------------------------------------------------------------------------
# this function counts the P[21] pairs
#
# It will loop through each matrix element and count its neighboring elements that
# are 1
#--------------------------------------------------------------------------------------
#--------------------------------------------------------------------------------------
A1<-A
pair_function<-function(){
  for ( i in 1:n){
    for (j in 1:n){
      if (A[i,j]==2){  
        if (i==n-1){
          if(A[n,j]==1){
            S21<-S21+1
          }
        } else if (A[(i+1)%%n,j]==1){
          S21<-S21+1
        }
        
        if ((i-1)==0){
          if (A[n,j]==1){
            S21<-S21+1
          }
        } else if (A[i-1,j]==1){
          S21<-S21+1
        }
        if (j==n-1){
          if(A[i,n]==1){
            S21<-S21+1
          }
        } else if (A[i,(j+1)%%n]==1){
          S21<-S21+1
        }
        
        if ((j-1)==0){
          if (A[i,n]==1){
            S21<-S21+1
          }
        } else if (A[i,j-1]==1){
          S21<-S21+1
        }
        
      }
    }
  }

  return(S21)

}  # end outside forloop

#--------------------------------------------------------------
#--------------------------------------------------------------
#--------------------------------------------------------------
# this is the main loop which runs the simulation
#--------------------------------------------------------------
#--------------------------------------------------------------
#--------------------------------------------------------------
timeseries = 0
for (k in 1:300){
 
  S21<-0               # this is the number of pairs of states of 2 and 1
  S21<-pair_function() # this is the function which counts pairs of 2 and 1
 
  S0<-sum(A[,]==0)     #sums the number of cells of state 0 at ith step
  S1<-sum(A[,]==1)     #sums the number of cells of state 1 "        "
  S2<-sum(A[,]==2)     #sums the number of cells of state 2 "        "
 
    
  tr<-(epsilon*S0)+(3/4)*(S21*phi)+(S2*mu)   #total rate
  timeseries = append(timeseries, timeseries[k] + rexp(1,tr))

  P0<-(epsilon*S0)/(tr)      # probability 1
  P1<-((3/4)*S21*phi)/(tr)   # probability 2
  P2<-(S2*mu)/(tr)           # probability 3
 
  #----------------------------------------------------------------------------------
  # the following cumulative probabilities generate the measure space for the events
  #-----------------------------------------------------------------------------------
  cump1<-P0      # the cumulative probability 1
  cump2<-P0+P1   # the cumulative probability 2
 

  #------------------------------------------------------------------------------------
  ## the following lines carry out the events -> pesticide decomposition, birth,
  ## and the application of pesticide/colony death
  ## IMPORTANT COMMANDS
  ##
  ## 1.) The "v" vector records the indices of the matrix elements(cells) at 
  ##     the state needed for the randomly selected event to occur
  ## 
  ##
  ## 2.) The "pointer" randomly samples from the indices of "v" to select only one 
  ##     element
  ##
  ##------------------------------------------------------------------------------------

  coin<-runif(1,0,1)    # sampling a number between 0 and 1 for the measure space

  if (coin < cump1){    # This is pesticide decomposition
  
    for (i in 1:n){
      for (j in 1:n){
        if (A[i,j]==0){
          v<-rbind(v,c(i,j))
        
        }
      }
    
    }  # end pesticide decomp event

    pointer<-sample(seq(1,length(v[,1]),1),1)
    A[v[pointer,1],v[pointer,2]]<-1
    v<-NULL

  } else if (coin >=cump1 && coin<=cump2){  ## This will give colony birth event
      for (i in 1:n){
        for (j in 1:n){
          if (A[i,j]==2){
            v<-rbind(v,c(i,j))
          }
        }
  
      } # end outside forloop statement 

      pointer<-sample(seq(1,length(v[,1]),1),1)
  
      if ((v[pointer,1])==n-1){
        if (A[n,v[pointer,2]]==1){
          A[n,v[pointer,2]]<-2
          v<-NULL
        } 
      } else if ((v[pointer,1]-1)==0){
        if (A[n,v[pointer,2]]==1){
          A[n,v[pointer,2]]<-2
          v<-NULL
        }
      } else if (v[pointer,2]==n-1){
        if (A[v[pointer,1],n]==1){
          A[v[pointer,1],n]<-2
          v<-NULL
        }
      } else if (v[pointer,2]==1){
        if (A[v[pointer,1],n]==1){
          A[v[pointer,1],n]<-2
          v<-NULL
        }
      } else if (A[(v[pointer,1]+1)%%n,v[pointer,2]]==1){
        A[(v[pointer,1]+1)%%n,v[pointer,2]]<-2
        v<-NULL
      } else if (A[(v[pointer,1]-1)%%n,v[pointer,2]]==1){
        A[(v[pointer,1]-1)%%n,v[pointer,2]]<-2
        v<-NULL
    
      } else if (A[v[pointer,1],(v[pointer,2]+1)%%n]==1){
        A[v[pointer,1],(v[pointer,2]+1)%%n]<-2
        v<-NULL
      } else if (A[v[pointer,1],(v[pointer,2]-1)%%n]==1){
        A[v[pointer,1],(v[pointer,2]-1)%%n]<-2
        v<-NULL
      }
  
  
  
  } else if ( coin > cump2){    ## pesticide is spray event
    for (i in 1:n){
      for (j in 1:n){
        if (A[i,j]==2){
          v<-rbind(v,c(i,j))
        }
      }
    }

    pointer<-sample(seq(1,length(v[,1]),1),1)
    A[v[pointer,1],v[pointer,2]]<-0
    v<-NULL
  } ## end the "else if" statement for pesticide spray event
  
  S0t<-c(S0t,S0) # record the number of pesticide states after one event
  S1t<-c(S1t,S1) #    ""                empty land states  ""
  S2t<-c(S2t,S2) #    ""                ant colony states  ""

} # end main for loop


 plot(timeseries,S0t/2500,ylim=c(0,1),col='green',type="l",lwd=.2,
      xlab="time (weeks)",ylab="state density")
      
 title("Cellular Automaton Simulation")
 lines(timeseries,S2t/2500,col='blue',xlab="time")
 lines(timeseries,S1t/2500,col='red')
# total_rate<-function(){  #sums each  the number of cells by "state"
\end{verbatim}



\end{document}